\documentclass[acmtocl]{acmsmall}

\usepackage{mathptmx,url}

\usepackage{proof,bussproofs}
\usepackage{amsfonts,amsmath,latexsym,amssymb}
\usepackage{tikz}
\usepackage{verbatim}
\usepackage{listings}
\usepackage{color}
\usepackage{multirow}
\usepackage{fancyvrb}
\usepackage{algpseudocode}
\usepackage{algorithm}

\usepackage{todonotes}

\usepackage{mathptmx}
\usepackage{comment}

\usepackage{tikz}
\usetikzlibrary{circuits.ee.IEC}

\usepackage{color}

\definecolor{LightGray}{gray}{0.8}

\newcommand{\Nat}{\mathbb{N}}

\newcommand{\kw}[1]{\ensuremath{\mathsf{#1}}}

\newcommand{\cat}[1]{\ensuremath{\mathbf{#1}}}

\newcommand\bcmdtab{\noindent\bgroup\tabcolsep=0pt%
  \begin{tabular}{@{}p{10pc}@{}p{20pc}@{}}}
\newcommand\ecmdtab{\end{tabular}\egroup}

\definecolor{LightGray}{gray}{0.8}

\begin{document}

\title{Structural Resolution: a Framework for Coinductive Proof Search and Proof Construction in Horn Clause Logic}
\author{Ekaterina Komendantskaya
\affil{School of Mathematical and Computer Sciences, Heriot-Watt University, UK}
Patricia Johann
\affil{Department of Computer Science, Appalachian State University, USA}
Martin Schmidt
\affil{University of Osnabr\"uck,
  Osnabr\"uck, Germany}
}

\begin{abstract}
Logic programming (LP) is a programming language based on first-order
Horn clause logic that uses SLD-resolution as a semi-decision
procedure. Finite SLD-computations are inductively sound and complete
with respect to least Herbrand models of logic programs. Dually, the
corecursive approach to SLD-resolution views infinite SLD-computations
as successively approximating infinite terms contained in programs'
greatest complete Herbrand models. State-of-the-art algorithms
implementing corecursion in LP are based on loop detection. However,
such algorithms support inference of logical entailment only for
rational terms, and they do not account for the important
property of productivity in infinite SLD-computations.  Loop detection
thus lags behind coinductive methods in interactive theorem proving
(ITP) and term-rewriting systems (TRS).
 
Structural resolution is a newly proposed alternative to
SLD-resolution that makes it possible to define and semi-decide a
notion of productivity appropriate to LP. 
In this paper, we prove soundness of structural resolution relative to Herbrand
model semantics for productive inductive, coinductive, and mixed
inductive-coinductive logic programs.

We introduce two algorithms that support coinductive proof search for infinite productive 
terms. One algorithm combines the method of loop detection with productive structural resolution, thus guaranteeing productivity 
of coinductive proofs for infinite rational terms. The other allows to make lazy sound observations of fragments of infinite irrational 
productive terms. 
This puts coinductive methods in LP on par with
productivity-based observational approaches to coinduction in ITP and
TRS.

\end{abstract}




\acmformat{Ekaterina Komendantskaya, et al., 2015.
Structural Resolution}

\begin{bottomstuff}
The work was supported by EPSRC grants EP/J014222/1 and EP/K031864/1
and NSF award 1420175.

\end{bottomstuff}


\maketitle

\section{Introduction}\label{sec:intro}
\subsection{A Symmetry of Inductive and Coinductive Methods}

Logic Programming (LP) is a programming language based on Horn clause
logic.  If $P$ is a logic program and $t$ is a (first-order) term,
then LP provides a mechanism for automatically inferring whether or
not $P$ logically entails $t$.  The traditional (inductive) approach
to LP is based on least fixed point semantics~\cite{Kow74,Llo87} of
logic programs, and defines, for every such program $P$, the
\emph{least Herbrand model} for $P$, i.e., the  set of
all (finite) ground terms {\em inductively entailed} by $P$.

\begin{example}\label{ex:nat} 
The program $P_1$ defines the set of natural numbers:

\vspace*{0.1in}

\noindent
0. $\mathtt{nat(0)} \; \gets \;$\\ 
1. $\mathtt{nat(s(X))} \; \gets \; \mathtt{nat(X)}$

\vspace*{0.1in}

\noindent 
The least Herbrand model for $P_1$ comprises the terms
$\mathtt{nat(0), \ nat(s(0)),}$ $\mathtt{nat(s(s(0))), \ldots}$
\end{example}

The clauses of $P_1$ can be viewed as inference rules
$\frac{}{\mathtt{nat(0)}}$ and
$\frac{\mathtt{nat(X)}}{\mathtt{nat(s(X))}}$, and the least Herbrand
model can be seen as the set obtained by the forward closure of these
rules. Some approaches to LP and first-order sequent calculi are based
on this inductive view~\cite{ba08,BrotherstonS11} of programs, which
is entirely standard~\cite{Sg12}. A similar view underlies inductive
type definitions in interactive theorem proving (ITP)~\cite{Agda,Coq}. For example,
$P_1$ also corresponds to the following Coq definition of an inductive
type:

\vspace*{0.1in}

\noindent
$\mathtt{Inductive \ nat \ : \ Type :=}$\\
$\hspace*{0.5in}\mathtt{| \ \ 0 \ : \ nat}$\\
$\hspace*{0.5in}\mathtt{| \ \ S \ : \ nat \rightarrow nat.}$

\vspace*{0.1in}

In addition to viewing logic programs inductively, we can also view
them coinductively. The \emph{greatest complete Herbrand model} for a
program $P$ takes the backward closure of the rules derived from $P$'s
clauses, thereby producing the largest set of finite and infinite
ground terms \emph{coinductively entailed} by $P$. For example, the
greatest complete Herbrand model for $P_1$ is the set containing all
of the finite terms in the least Herbrand model for $P_1$, together
with the term $\mathtt{nat(s(s(...)))}$ representing the first limit
ordinal. The coinductive view of logic programs corresponds to
coinductive type definitions in ITP.

As it turns out, some logic programs have no natural inductive
semantics and should instead be interpreted coinductively:

\begin{example}
\label{ex:natstream} 
The program $P_2$ defining streams of natural numbers comprises the
clauses of $P_1$ and the following additional one:

\vspace*{0.1in}

\noindent 
2. $\mathtt{nats(scons(X,Y))} \; \gets \; \mathtt{nat(X)},\mathtt{nats(Y)}$ 

\vspace*{0.1in}

\noindent 
No terms defined by $\mathtt{nats}$ are contained in the least
Herbrand model for $P_2$, but $P_2$'s greatest complete Herbrand model
contains infinite terms representing infinite streams of natural numbers, like e.g. the infinite term $t = \mathtt{nats(scons(0,scons(0,
  \ldots)}$ representing the infinite stream of zeros.
\end{example}

\noindent
The program $P_2$ corresponds to the following Coq definition of a
coinductive type:

\vspace*{0.1in}

\noindent
$\mathtt{CoInductive \  nats \ : Type  :=}$\\
$\hspace*{0.5in}\mathtt{SCons : nat \rightarrow nats \rightarrow nats.}$

\vspace*{0.1in}

\noindent 
The formal relation between logic programs, Herbrand models, and types
is analysed in \cite{HeintzeJ92}.

\subsection{Preconditions for an Operational Semantics?}

The (least and greatest complete) Herbrand models for programs, as
defined by (forward and backward rule closure, respectively, of) their
clauses, provide one important way to understand logic programs. But
an equally important way is via their computational behaviours.
Rather than using Herbrand models to give meaning to ``inductive'' and
``coinductive'' logic programs, we can also use the operational
properties of SLD-resolution to assign programs semantics that take
into account the computational behaviours that deliver those
models. Ideally, we would like to do this in such a way that the
symmetry between the Herbrand model interpretations of inductive and
coinductive programs as the sets of terms (i.e., the types) they
define is preserved by these computational interpretations.

The transition from types to computations is natural in ITP, where
recursive functions consume inputs of inductive types and, dually,
corecursive functions produce outputs of coinductive types. Since
systems such as Coq and Agda require recursive functions to be
terminating in order to be sound, and since SLD-resolution similarly
requires a logic program's derivations to be terminating in order for
them to be sound with respect to that program's least Herbrand model,
we might dually expect a logic program's non-terminating derivations
to compute terms in its greatest complete Herbrand model.  However,
non-termination does not play a role for coinduction dual to that
played by termination for induction. In particular, the fact that a
logic program admits non-terminating SLD-derivations does not, on its
own, guarantee that the program's computations completely capture its
greatest complete Herbrand model:

\begin{example}\label{ex:bad}
The following ``bad'' program gives rise to an infinite SLD-derivation:

\vspace*{0.1in}

\noindent
0. $\mathtt{bad(f(X))} \; \gets \; \mathtt{bad(f(X))}$

\vspace*{0.1in}

\noindent 
Although this program does not compute any infinite terms, the
infinite term $\mathtt{bad(f(f(...)))}$ is in its greatest complete
Herbrand model.
\end{example}

It is important to note that the ``badness'' of this program is
unrelated to the fact that LP is untyped.  The following corecursive
function is equally ``bad'', and will be rejected by Coq:

\vspace*{0.1in}

\noindent
$\mathtt{CoInductive \ Stream \ A : Type :=}$\\ $\mathtt{\ \ \ \ \
SCons : A \rightarrow Stream \ A \rightarrow Stream \ A.}$\\[1ex]
$\mathtt{CoFixpoint\ bad \ (f: \ A \ \rightarrow \ A)\ (x:A) \ : \
Stream \ A := } \ $ $\mathtt{ bad \ f \ (f \ x). }$

\vspace*{0.1in}

The problem here actually lies in the fact that both the LP and the
ITP versions of the above ``bad'' program fail to satisfy the
important property of productivity. The productivity requirement on
corecursive programs for systems such as Coq and Agda reflects the
fact that an infinite computation can only be consistent with its
intended coinductive semantics if it is {\em globally productive},
i.e., if it actually produces an infinite object in the limit. But in
order to give an operational meaning to ``in the limit'' --- which is
not itself a computationally tractable concept --- productivity is
usually interpreted in terms of finite observability. Specifically, a
function can be (finitely) observed to be globally productive if each
part of its infinite output can be generated in finite time. We call
this kind of productivity {\em observational productivity}.  A similar
notion of an observationally productive infinite computation has also
been given for stream productivity in term rewriting systems (TRS)~\cite{EndrullisGHIK10,EndrullisHHP015}.
Moreover, a variety of syntactic guardedness checks have been
developed to semi-decide observational productivity in ITP in
practice~\cite{Coq94,Gimenez98,BK08}. However, prior to~\cite{KJS17}, LP did not
have any notion of an observationally productive program, and
therefore did not have a corresponding operational semantics based on
any such notion.

\subsection{Symmetry Broken}

It is well-known that termination captures the least Herbrand model
semantics of (inductive) logic programs computationally: the
terminating and successful SLD-derivations for any program $P$ give a
decision procedure for membership in the least Herbrand model for
$P$. For example, after a finite number of SLD-derivation steps we can
conclude that $\mathtt{nat(X)}$ is in the least Herbrand model for
program $P_1$ if $\mathtt{X} = \mathtt{0}$. Termination of
SLD-derivations thus serves as a computational precondition for
deciding logical entailment.

But for programs, like $P_2$, that admit non-terminating derivations,
SLD-resolution gives only a {\em semi}-decision procedure for logical
entailment. Indeed, if an SLD-derivation for a program and a query
terminates with success, then we definitely know that the program
logically entails the term being queried, and thus that this term is
in the greatest complete Herbrand model for the program. But if an
SLD-derivation for the program and query does not terminate, then we
can infer nothing. It is therefore natural to ask:

\vspace*{0.1in}\noindent {\em Question: Is it possible to capture the
  greatest complete Herbrand model semantics for potentially
  non-terminating logic programs computationally? If so, how?}

\vspace*{0.1in}\noindent That is, can we restore the symmetry between
terminating and potentially non-terminating logic programs so that
that the correspondence between a terminating program's Herbrand
semantics and its computational behaviour also holds for
non-terminating programs?

In one attempt to match the greatest complete Herbrand semantics for
potentially non-terminating programs, an operational counterpart ---
called \emph{computations at infinity} --- was introduced in the
1980s~\cite{Llo87,EmdenA85}. The operational semantics of a
potentially non-terminating logic program $P$ was then taken to be the
set of all infinite ground terms computable by $P$ at infinity. For
example, the infinite ground term $t$ in Example~\ref{ex:natstream} is
computable by $P_2$ at infinity starting with the query $? \gets
\mathtt{nats(X)}$. Although computations at infinity do better capture
the computational behaviour of non-terminating logic programs, they
are still only sound, and not complete, with respect to those
programs' greatest complete Herbrand models. For example, the infinite
term $\mathtt{bad(f(f(...)))}$ is in the greatest complete Herbrand
model for the ``bad'' program of Example~\ref{ex:bad}, as noted there,
but is not computable at infinity by that program.

Interestingly, computations at infinity capture the same intuition
about globally productive infinite SLD-derivations that underlies the
productivity requirement for corecursive functions in ITP~\cite{Coq94,Gimenez98,BK08} and productive streams in
TRS~\cite{EndrullisGHIK10,EndrullisHHP015}. That is, they insist that each infinite SLD-derivation actually
produces an (infinite) term.  This observation leads us to adapt the
terminology of~\cite{Llo87,EmdenA85} and say that a logic program $P$
is {\em SLD-productive} if every infinite SLD-derivation for $P$
computes an infinite term at infinity. SLD-productivity captures the
difference in computational behavior between programs, like $P_2$,
that actually do compute terms at infinity, from ``bad'' programs,
like that of Example~\ref{ex:bad}, that do not. While computations at
infinity are not complete with respect to greatest complete Herbrand
models for non-SLD-productive logic programs, for SLD-productive
programs they {\em are}. For example, the SLD-productive program below
is similar to our non-SLD-productive ``bad'' program and its greatest
complete Herbrand model is computed in the same way:

\vspace*{0.1in}
\noindent
0. $\mathtt{good(f(X))} \; \gets \; \mathtt{good(X)}$

\vspace*{0.1in}

\noindent
But because this program is SLD-productive --- and, therefore,
``good'' --- the infinite term $\mathtt{good(f(f(...)))}$
corresponding to the problematic term above is not only in its
greatest complete Herbrand model, but is also computable at infinity.

In light of the above, we concentrate on productive logic programs,
shifting our focus away from greatest complete Herbrand models and
toward computations at infinity, to give such programs a more
computationally relevant semantics. But a big challenge still remains:
even for productive programs, the notion of computations at infinity
does not by itself give rise to implementations. Specifically,
although SLD-productivity captures the important requirement that
infinite computations actually produce output, it does not give a
corresponding notion of finite observability, as ITP and TRS stream productivity approaches
productivity do. We therefore refine our question above to ask:

\vspace*{0.1in}

\noindent {\em Question (refined): Can we formulate a computational
  semantics for LP that redefines productivity in terms of finite
  observability, as is done elsewhere in the study of programming
  languages, and that does this in such a way that it both yields
  implementations and ensures soundness and completeness with respect
  to computations at infinity (rather than greatest complete Herbrand
  models)? If so, how?}

\vspace*{0.1in}

Thirty years after the initial investigations into coinductive
computations, coinductive logic programming, implemented as CoLP, was
introduced~\cite{GuptaBMSM07,SimonBMG07}.  CoLP provides practical
methods for terminating infinite SLD-derivations. 
CoLP's coinductive proof search is based on a
loop detection mechanism that requires the programmer to supply
annotations marking every predicate as either inductive or
coinductive. For coinductive predicates, CoLP observes finite
fragments of SLD-derivations, checks them for unifying subgoals, and
terminates when loops determined by such subgoals are found. A similar
loop detection method is employed for type class inference in the
Glasgow Haskell Compiler (GHC)~\cite{Lammel:2005}, and CoLP itself is
used for type class inference in Featherweight
Java~\cite{AnconaLagorio11}.

\vspace*{-0.1in}

\begin{example}\label{ex:natstream2} 
If $\mathtt{nats}$ is marked as coinductive in $P_2$, then the query
$? \gets \mathtt{nats(X)}$ gives rise to an SLD-derivation with a
sequence of subgoals $\mathtt{nats(X)} \leadsto^{X \mapsto scons(0,Y')} \mathtt{nats(Y')} \leadsto \ldots$.
Observing that
$\mathtt{nats(X)}$ and $\mathtt{nats(Y')}$ unify
and thus comprise a loop, CoLP concludes that $\mathtt{nats(X)}$ has
been proved 
and  returns the
answer $\mathtt{X=scons(0,X)}$ in the form of a ``circular'' term
indicating that $P_2$ logically entails the term $t$ in
Example~\ref{ex:natstream}.
\end{example}

CoLP is sound, but incomplete, relative to greatest complete Herbrand
models~\cite{GuptaBMSM07,SimonBMG07}. But, perhaps surprisingly, it is
{\em neither} sound {\em nor} complete relative to computations at
infinity. CoLP is not sound because our ``bad'' program from
Example~\ref{ex:bad} computes no infinite terms at infinity for the
query $? \gets \mathtt{bad(X)}$, whereas CoLP notices a loop and
reports success (assuming the predicate $\mathtt{bad}$ is marked as
coinductive). CoLP is not complete because not all terms computable at
infinity by all programs can be inferred by CoLP. In fact, CoLP's loop
detection mechanism can only terminate if the term computable at
infinity is a \emph{rational}
term~\cite{Courcelle83,JaffarS86}. Rational terms are terms that can
be represented as trees that have a finite number of distinct
subtrees, and can therefore be expressed in a closed finite form
computed by circular unification. The ``circular'' term $\mathtt{X}$ =
$\mathtt{scons(0,X)}$ in Example~\ref{ex:natstream2} is so expressed.
For irrational terms, CoLP simply does not terminate:

\begin{example}\label{ex:add}\label{ex:fs}
The program $P_3$ defines addition on the Peano numbers, together with
the stream of Fibonacci numbers:

\vspace*{0.05in}\noindent
0. $\mathtt{add(0,Y,Y)} \; \gets \;$ \\
1. $\mathtt{add(s(X),Y,s(Z)) } \; \gets \; \mathtt{add(X,Y,Z)}$\\
2. $\mathtt{fibs(X,Y,cons(X,S)) } \; \gets \; \mathtt{add(X,Y,Z),
  fibs(Y,Z,S)}$\\

\vspace*{-0.15in}

\vspace*{0.05in}\noindent From a coinductive perspective, $P_3$ is
semantically and computationally meaningful. It computes the infinite
term $t^* = \mathtt{fibs(0,s(0),cons(0,cons(s(0),cons(s(0),
  cons(s(s(0)), \ldots))))}$, and thus the stream of Fibonacci numbers
(in the third argument to $\mathtt{fibs}$). The term $t^*$ is both
computable at infinity by $P_3$ and contained in $P_3$'s greatest
complete Herbrand model.  Nevertheless, when CoLP processes the
sequence $\mathtt{fibs(0,s(0),cons(0,S))}, \,
\mathtt{fibs(s(0),s(0),}$ $\mathtt{cons(s(0),S'))}$,
$\mathtt{fibs(s(0), s(s(0),cons(s(0),S''))},\, \ldots$ of subgoals for
the program $P_3$ and query $? \gets \mathtt{fibs(0,s(0),X)}$ giving
rise to $t^*$, it cannot unify any two of them, and thus does not
terminate.
\end{example}

The upshot is that CoLP cannot faithfully capture the operational
meaning of computations at infinity.  

\subsection{Structural Resolution for Productivity}
It has been strongly argued in~\cite{KPS12-2,KJS17,FK16} that the recently discovered \emph{structural resolution} can 
  help to combine the intuitive notion of computations at infinity and the coinductive reasoning \emph{\'a la} CoLP.
We explain the main idea behind structural resolution by means of an example.


\begin{example}\label{ex:from}
The coinductive program $P_4$ has the single clause

\vspace*{0.1in}

\noindent
0. $\mathtt{from(X, scons(X,Y))} \gets \mathtt{from(s(X),Y)}$

\vspace*{0.1in}

\noindent Given the query $? \gets \mathtt{from(0, X)}$, and writing
$\mathtt{[\_,\_]}$ as an abbreviation for the stream constructor
$\mathtt{scons}$ here, we have that the infinite term $t' =
\mathtt{from(0,[0,[s(0),[s(s(0)),\ldots]]])}$ is computable at
infinity by $P_4$ and is also contained in the greatest Herbrand model
for $P_4$. By the same argument as in Examples~\ref{ex:natstream2}
and~\ref{ex:fs}, coinductive reasoning on this query cannot be handled
by the loop detection mechanism of CoLP because the term $t'$ is
irrational.

Structural resolution allows us to separate the infinite derivation steps that compute $t'$ at infinity
into term rewriting and unification steps as shown below,
with term rewriting steps shown vertically and unification steps shown
horizontally. This separation makes it easy to see that $P_4$ is {\em
finitely observable}, in the sense that all of its derivations by term
rewriting alone terminate.

\begin{center}
\begin{tikzpicture}[scale=0.30,baseline=(current bounding box.north),grow=down,level distance=20mm,sibling distance=50mm,font=\footnotesize]
  \node { $\mathtt{from(0,X)}$};
  \end{tikzpicture}
$\stackrel{\{\mathtt{X} \mapsto [\mathtt{0},\mathtt{X}']\}}{\rightarrow}$
\begin{tikzpicture}[scale=0.30,baseline=(current bounding box.north),grow=down,level distance=20mm,sibling distance=60mm,font=\footnotesize ]
  \node { $\mathtt{from(0,[0, X'])}$}
          child { node {$\mathtt{from(s(0),X')}$}};
  \end{tikzpicture}
	$\stackrel{\{\mathtt{X}'\mapsto [\mathtt{s(0)},\mathtt{X}'']\}}{\rightarrow}$
	\begin{tikzpicture}[scale=0.30,baseline=(current bounding box.north),grow=down,level distance=20mm,sibling distance=60mm,font=\footnotesize ]
  \node { $\mathtt{from(0,[0, [s(0), X'']])}$}
	child { node{ $\mathtt{from(s(0),[s(0),X''])}$ }
          child { node {$\mathtt{from(s(s(0)),X'')}$}}};
  \end{tikzpicture}~
$\stackrel{\{\mathtt{X}'' \mapsto [\mathtt{s(s(0))}, \mathtt{X}''']\}}{\rightarrow}$
\end{center}
\end{example}

It is intuitively pleasing to represent sequences of term rewriting
reductions as trees. We call these \emph{rewriting trees} to mark
their resemblance to TRS~\cite{Terese}.  Full SLD-derivation steps can
be represented by transitions between rewriting trees. These
transitions are determined by most general unifiers of rewriting tree
leaves with program clauses; see Section~\ref{sec:cotrees} below, as
well as~\cite{JohannKK15}, for more detail. In \cite{FK15,JohannKK15}
this method of separating SLD-derivations into rewriting steps and
unification-driven steps is
called \emph{structural resolution}, or {\em S-resolution} for short.

With S-resolution in hand, we can define a logic program to be {\em
  (observationally) productive} if it is finitely observable, i.e., if
all of its rewriting trees are finite. Our ``good'' program above is
again productive, whereas the ``bad'' one is not --- but now
``productive'' means productive in this new observational sense.
Productivity in S-resolution corresponds to termination in TRS, see \cite{FK15}.
In addition, \cite{KJS17} presents 
an algorithm and an implementation that semi-decides observational 
productivity.

\emph{One question remains: if  termination is an effective pre-condition for semi-deciding inductive soundness in LP, 
can observational productivity in S-resolution become a similarly effective pre-condition for reasoning about computations at infinity?}

To answer this question, we present a complete study of inductive and coinductive properties of S-resolution, and
establish the following results:

\begin{itemize}
	\item S-resolution is inductively sound and complete relative to the least Herbrand model semantics;
	\item S-resolution is coinductively sound relative to the greatest Herbrand model semantics;
	\item Infinite observationally productive computations by S-resolution are sound and complete relative to SLD-computations at infinity;
\end{itemize}

The above results prove that indeed the notion of observational productivity simplifies reasoning about global productivity (given by SLD-computations at infinity).
 
 

However, thanks to separation of rewriting and substitution steps, S-resolution can be soundly used to lazily observe finite fragments of infinite irrational terms. We introduce an algorithm for such sound observations, we attach an implementation to this paper:~\url{https://github.com/coalp}.


\subsection{Paper overview}

The paper proceeds as follows.
In
Section~\ref{sec:cotrees} we introduce background definitions
concerning LP, including least and greatest complete Herbrand model
semantics and operational semantics of SLD- and S-resolution given by
reduction systems.  In Section~\ref{sec:semantics}, we prove the
soundness, and show the incompleteness, of S-resolution reductions
with respect to least Herbrand models.  In Section~\ref{sec:ST}, we
regain completeness of proof search by introducing rewriting trees and
rewriting tree transitions (which we call \emph{S-derivations}), and
proving the soundness and completeness of successful S-derivations
with respect to least Herbrand models. This completes the discussion
of inductive properties of proof-search by S-resolution, and lays the
foundation for developing a new coinductive operational semantics for
LP via S-resolution in Sections~\ref{sec:models} and~\ref{sec:cc}. In
Section~\ref{sec:models}, we define S-computations at infinity and
show that they are sound and complete relative to SLD-computations at infinity. We reconstruct
the standard soundness result for computations at infinity relative to
greatest complete Herbrand models, but now for S-computations at
infinity. 

In Section~\ref{sec:concl} we conclude and discuss related work.

\section{Preliminaries}\label{sec:cotrees}

In this section we introduce Horn clauses, and recall the declarative
(big-step) semantics of logic programs given by least and greatest
complete Herbrand models. We also introduce structural resolution
reduction by means of an operational (small-step) semantics. To enable
the analysis of coinductive semantics and infinite terms, we adopt the
standard view of terms as trees~\cite{Courcelle83,JaffarS86,Llo87}.

\subsection{First-Order Signatures, Terms, Clauses}

We write $\Nat^*$ for the set of all finite words over the set $\Nat$
of natural numbers. The length of $w\in\Nat^*$ is denoted $|w|$. The
empty word $\epsilon$ has length $0$; we identify $i \in \Nat$ and the
word $i$ of length $1$.  Letters from the end of the alphabet denote
words of any length, and letters from the middle of the alphabet
denote words of length $1$. The concatenation of $w$ and $u$ is
denoted $wu$; $v$ is a {\em prefix} of $w$ if there exists a $u$ such
that $w = vu$, and a {\em proper prefix} of $w$ if $u \not =
\epsilon$.

A set $L \subseteq \Nat^*$ is a \emph{(finitely branching) tree
  language} provided: i) for all $w \in \Nat^*$ and all $i,j \in
\Nat$, if $wj \in L$ then $w \in L$ and, for all $i<j$, $wi \in L$;
and ii) for all $w \in L$, the set of all $i\in \Nat$ such that $wi\in
L$ is finite.  A non-empty tree language always contains $\epsilon$,
which we call its {\em root}. The \emph{depth} of a tree language $L$
is the maximum length of a word in $L$.  A tree language is {\em
  finite} if it is a finite subset of $\Nat^*$, and {\em infinite}
otherwise. A word $w \in L$ is also called a {\em node} of $L$. If $w
= w_0w_1...w_l$ then $w_0w_1...w_k$ for $k < l$ is an {\em ancestor}
of $w$. The node $w$ is the \emph{parent} of $wi$, and nodes $wi$ for
$i \in \Nat$ are {\em children} of $w$.  A \emph{branch} of a tree
language $L$ is a subset $L'$ of $L$ such that, for all $w,v\in L'$,
$w$ is an ancestor of $v$ or $v$ is an ancestor of $w$. If $L$ is a
tree language and $w$ is a node of $L$, the \emph{subtree of $L$ at
  $w$} is $L \backslash w = \{v \mid wv \in L\}$.

A \emph{signature} $\Sigma$ is a non-empty set of \emph{function
  symbols}, each with an associated arity. The arity of $f \in \Sigma$
is denoted $\mathit{arity}(f)$. To define terms over $\Sigma$, we
assume a countably infinite set $\mathit{Var}$ of {\em variables}
disjoint from $\Sigma$, each with arity $0$. Capital letters from the
end of the alphabet denote variables in $\mathit{Var}$.  If $L$ is a
non-empty tree language and $\Sigma$ is a signature, then a
\emph{term} over $\Sigma$ is a function $t: L \rightarrow \Sigma \cup
\mathit{Var}$ such that, for all $w\in L$, $\mathit{arity}(t(w)) =
\;\,\mid \!\{i \mid wi \in L\}\!\mid$.  Terms are finite or infinite
according as their domains are finite or infinite.  A term $t$ has a
depth $\mathit{depth}(t) = \max\{|w| \mid \, w\in L\}$.  The subtree
$\mathit{subterm}(t,w)$ of $t$ at node $w$ is given by $t': (L
\backslash w) \rightarrow \Sigma \cup Var$, where $t'(v) = t(wv)$ for
each $wv\in L$. The set of finite (infinite) terms over a
signature $\Sigma$ is denoted by $\mathbf{Term}(\Sigma)$
($\mathbf{Term}^\infty(\Sigma)$). The set of {\em all} (i.e., finite
{\em and} infinite) terms over $\Sigma$ is denoted by
$\mathbf{Term}^\omega(\Sigma)$. Terms with no occurrences of variables
are \emph{ground}. We write $\textbf{GTerm}(\Sigma)$
($\textbf{GTerm}^\infty(\Sigma)$, $\mathbf{GTerm}^\omega(\Sigma)$) for
the set of finite (infinite, {\em all}) ground terms over $\Sigma$.

A \emph{substitution} over $\Sigma$ is a total function $\sigma:
\mathit{Var} \to \mathbf{Term}^{\omega}(\Sigma)$. 
Substitutions are extended from variables
to terms homomorphically: if $t\in \mathbf{Term}(\Sigma)$ and $\sigma
\in \mathbf{Subst}^\omega(\Sigma)$, then the {\em application}
$\sigma(t)$ is $(\sigma(t))(w) = t(w)$ if $t(w) \not \in
\mathit{Var}$, and $(\sigma(t))(w) = (\sigma(\mathtt{X}))(v)$ if $w =
uv$, $t(u) = \mathtt{X}$, and $\mathtt{X} \in \mathit{Var}$.

We say that $\sigma$
is a {\em grounding substitution for $t$} if $\sigma(t) \in
\textbf{GTerm}^{\omega}(\Sigma)$, and is just a {\em grounding
  substitution} if its codomain is $\textbf{GTerm}^{\omega}(\Sigma)$.
We write $\mathit{id}$ for the identity substitution. The set of all
substitutions over a signature $\Sigma$ is
$\mathbf{Subst}^\omega(\Sigma)$ and the set of all substitutions over
$\Sigma$ with only finite terms in their codomains is
$\mathbf{Subst}(\Sigma)$.  
Composition of substitutions is denoted by juxtaposition.  Composition
is associative, so we write $\sigma_3\sigma_2\sigma_1$ rather than
$(\sigma_3\sigma_2)\sigma_1$ or $\sigma_3(\sigma_2\sigma_1)$.

A substitution $\sigma \in \cat{Subst}(\Sigma)$ is a \emph{unifier}
for $t, u \in \cat{Term}(\Sigma)$ if $\sigma(t) = \sigma(u)$, and is a
\emph{matcher} for $t$ against $u$ if $\sigma(t) = u$.  If $t, u \in
\cat{Term}^\omega(\Sigma)$, then we say that $u$ is an {\em instance}
of $t$ if $\sigma(t) = u$ for some $\sigma \in
\cat{Subst}^\omega(\Sigma)$; note that if $t,u \in
\cat{Term}(\Sigma)$, i.e., if $t$ and $u$ are finite terms, then the
codomain of $\sigma$ can be taken, without loss of generality, to
involve only finite terms.  A substitution $\sigma_1\in
\cat{Subst}(\Sigma)$ is {\em more general} than a substitution
$\sigma_2 \in \cat{Subst}(\Sigma)$
if there exists a substitution $\sigma \in \cat{Subst}^\omega(\Sigma)$
such that $\sigma \sigma_1(\mathtt{X}) = \sigma_2(\mathtt{X})$ for
every $\mathtt{X} \in \mathit{Var}$. A substitution $\sigma \in
\cat{Subst}(\Sigma)$ is a {\em most general unifier} ({\em mgu}) for
$t$ and $u$, denoted $t \sim_\sigma u$, if it is a unifier for $t$ and
$u$ and is more general than any other such unifier. A {\em most
  general matcher} ({\em mgm}) $\sigma$ for $t$ against $u$, denoted
$t \prec_\sigma u$, is defined analogously. Both mgus and mgms are
unique up to variable renaming if they exist.  Unification is
reflexive, symmetric, and transitive, but matching is reflexive and
transitive only. Mgus and mgms are computable by Robinson's seminal
unification algorithm~\cite{Rob63}.

In many unification algorithms, the \emph{occurs check} condition is
imposed, so that substitution bindings of the form $\mathtt{X} \mapsto
t(\mathtt{X})$, where $t(\mathtt{X})$ is a term containing
$\mathtt{X}$, are disallowed. In this case, mgus and mgms can always
be taken to be {\em idempotent}, i.e., such that the sets of variables
appearing in their domains and codomains are disjoint. The occurs
check is critical for termination of unification algorithms, and this
is, in turn, crucial for the soundness of classical SLD-resolution;
see below.

In logic programming, a clause $C$ over $\Sigma$ is a pair $(A,
[B_0,...,B_n])$, where $A \in \mathbf{Term}(\Sigma)$ and $[B_0, \ldots
  B_n]$ is a list of terms in $\mathbf{Term}(\Sigma)$; such a clause
is usually written $A \gets B_0, \ldots , B_n$.  Note that the list of
terms can be the empty list $[\;]$. We will identify the singleton
list $[t]$ with the term $t$ when convenient.  The head $A$ of $C$ is
denoted $\mathit{head}(C)$ and the body $B_0, \ldots , B_n$ of $C$ is
denoted $\mathit{body}(C)$. A {\em goal clause} $G$ over $\Sigma$ is a
clause $? \gets B_0, \ldots, B_n$ over $\Sigma \cup \{?\}$, where $?$
is a special symbol not in $\Sigma \cup \mathit{Var}$. We abuse
terminology and consider a goal clause over $\Sigma$ to be a clause
over $\Sigma$. The set of all clauses over $\Sigma$ is denoted by
$\mathbf{Clause}(\Sigma)$.  A \emph{logic program} over $\Sigma$ is a
total function from a set $\{0,1,\dots,n\} \subseteq \Nat$ to the set
of non-goal clauses over $\Sigma$. The clause $P(i)$ is called the
$i^{th}$ clause of $P$. If a clause $C$ is $P(i)$ for some $i$, we
write $C \in P$.  The set of all logic programs over $\Sigma$ is
denoted $\mathbf{LP}(\Sigma)$.  

The \emph{predicate} of a clause $C$
is the top symbol of the term $\mathit{head}(C)$. The predicates of a
program are the predicates of its clauses.  We assume that all logic 
programs are written within first-order Horn logic, with proper syntactic 
checks implied on the predicate position. 
If this assumption is made, the algorithm of SLD-resolution as well as other alternative 
algorithms we consider in the following sections do not introduce any syntactic inconsistencies 
to the operational semantics. 

The {\em arity} of $P\in
\mathbf{LP}(\Sigma)$ is the number of clauses in $P$, i.e., is
$|\mathit{dom}(P)|$, and is denoted $\mathit{arity}(P)$. The {\em
  arity} of $C\in \mathbf{Clause}(\Sigma)$ is $|\mathit{body}(C)|$,
and is similarly denoted $\mathit{arity}(C)$.

We extend substitutions from variables to clauses and programs
homomorphically. We omit these standard definitions. The variables of
a clause $C$ can be renamed with ``fresh'' variables to get an
$\alpha$-equivalent clause that is interchangeable with $C$. We assume
variables have been renamed when convenient. This is standard and
helps avoid circular (non-terminating) unification and matching.

\subsection{Big-step Inductive and Coinductive Semantics for LP}

We recall the least and greatest complete Herbrand model constructions
for LP~\cite{Llo87}. We express the definitions in the form of a
big-step semantics for LP, thereby exposing duality of inductive and
coinductive semantics for LP in the style of~\cite{Sg12}. We start
by giving inductive interpretations to logic programs.

\begin{definition}\label{df:irules}
Let $P\in \cat{LP}(\Sigma)$. The {\em big-step rule for $P$} is given
by
\[\frac{P \models \sigma(B_1), \,\ldots \, ,P \models \sigma(B_n)}{P
  \models \sigma(A)}\] where $A \gets B_1, \ldots B_n$ is a clause in
$P$ and $\sigma \in \mathbf{Subst}^{\omega}(\Sigma)$ is a grounding
substitution.
\end{definition}

\noindent
Following standard terminology (see, e.g.,~\cite{Sg12}), we say that
an inference rule is \emph{applied forward} if it is applied from top
to bottom, and that it is \emph{applied backward} if it is applied
from bottom to top. If a set of terms is closed under forward
(backward) application of an inference rule, we say that it is {\em
  closed forward} (resp., {\em closed backward}) under that rule.  If
the $i^{th}$ clause of $P \in \textbf{LP}(\Sigma)$ is involved in an
application of the big-step rule for $P$, then we may say that we have
applied the {\em big-step rule for $P(i)$}.

\begin{definition}\label{def:model}
The {\em least Herbrand model} for $P \in \textbf{LP}(\Sigma)$ is the
smallest set $M_P \subseteq \cat{GTerm}(\Sigma)$ that is closed
forward under the big-step rule for $P$.
\end{definition}

\begin{example}
The least Herbrand model for $P_1$ is 
$\{\mathtt{nat(0)},$ $\mathtt{nat(s(0))},$ $\mathtt{nat(s(s(0)))},
\ldots\}$.
\end{example}

The requirement that $M_P\subseteq \cat{GTerm}(\Sigma)$ entails that
only ground substitutions $\sigma \in \cat{Subst}(\Sigma)$ are used in
the forward applications of the big-step rule involved in the
construction of $M_P$. Next we give coinductive interpretations to
logic programs. For this we do not impose any finiteness requirement
on the codomain terms of $\sigma$.

\begin{definition}\label{def:cmodel}
The {\em greatest complete Herbrand model} for $P \in
\textbf{LP}(\Sigma)$ is the largest set $M^{\omega}_P \subseteq
\cat{GTerm}^{\omega}(\Sigma)$ that is closed backward under the
big-step rule for $P$.
\end{definition}

\begin{example}
The greatest complete Herbrand model for $P_1$ is $\{\mathtt{nat(0)},$
$\mathtt{nat(s(0))},$ $\mathtt{nat(s(s(0)))},$ $\ldots\} \; \bigcup \;
\{\mathtt{nat(s(s(...)))}\}$.  Indeed, there is an infinite inference
for $\mathtt{nat(s(s(...)))}$ obtained by repeatedly applying the
big-step rule for $P_1(1)$ backward.
\end{example}

Definitions~\ref{def:model} and~\ref{def:cmodel} could alternatively
be given in terms of least and greatest fixed point operators, as in,
e.g.,~\cite{Llo87}. To ensure that $\cat{GTerm}(\Sigma)$ and
$\cat{GTerm}^{\omega}(\Sigma)$ are non-empty, and thus that the least
and greatest Herbrand model constructions are as intended, it is
standard in the literature to assume that $\Sigma$ contains at least
one function symbol of arity $0$.
We will make this assumption throughout the remainder of this paper.

\subsection{Small-step Semantics for LP}

Following~\cite{FK15}, we distinguish the following three kinds of
reduction for LP.

\begin{definition}
If $P \in \textbf{LP}(\Sigma)$ and $t_1, \ldots, t_n
\in \mathbf{Term}(\Sigma)$, then
\begin{itemize}
\item {\em SLD-resolution reduction}: $[t_1, \ldots, t_i, \ldots ,
  t_n] \leadsto_P [\sigma(t_1), \ldots, \sigma(t_{i-1}), \sigma(B_0),
  \ldots \sigma(B_m), \sigma(t_{i+1}), \ldots, \sigma(t_n)]$ if $A
  \gets B_0 , \ldots , B_m \in P$ and $t_i \sim_{\sigma} A$.

\vspace*{0.1in}

\item {\em rewriting reduction}: $[t_1, \ldots , t_i , \ldots ,
  t_n] \rightarrow_P [t_1, \ldots, t_{i-1},\sigma(B_0), \ldots
  \sigma(B_m), t_{i+1}, \ldots, t_n]$ if $A \gets B_0 , \ldots , B_m
  \in P$ and $A \prec_{\sigma} t_i$.

\vspace*{0.1in}

\item {\em substitution reduction}: $[t_1, \ldots , t_i , \ldots,
  t_n] \hookrightarrow_P [\sigma(t_1), \ldots, \sigma(t_i), \ldots ,
  \sigma(t_n)]$ if $A \gets B_0 , \ldots , B_m \in P$ and $t_i
  \sim_{\sigma} A$.
\end{itemize}
We assume, as is standard in LP, that all variables are \emph{renamed apart} when terms are matched or unified against the program clauses.
\end{definition}
	
If $r$ is any reduction relation, we will abuse terminology and call
any (possibly empty) sequence of $r$-reduction steps an {\em
  $r$-reduction}. When there exists no list $L$ of terms such that
$[t_1, \ldots , t_i , \ldots , t_n] \rightarrow_P L$ we say that
$[t_1, \ldots, t_n]$ is in {\em $\rightarrow$-normal form} with
respect to $P$.  We write $[t_1,...,t_n] \rightarrow^{\mu}_P$ to
indicate the reduction of $[t_1,...,t_n]$ to its $\rightarrow$-normal
form with respect to $P$ if this normal form exists, and to indicate
an infinite reduction of $[t_1,...,t_n]$ with respect to $P$
otherwise. We write $\rightarrow^n$ to denote rewriting by {\em at
  most} $n$ steps of $\rightarrow$, where $n$ is a natural number. We
will use similar notations for $\leadsto$ and $\hookrightarrow$ as
required.  Throughout this paper we may omit explicit mention of $P$
and/or suppress $P$ as a subscript on reductions when it is clear from
context.

We are now in a position to define the structural resolution reduction,
also called the {\em S-resolution reduction} for short. We have:

\begin{definition}
For $P \in \textbf{LP}(\Sigma)$, the {\em structural resolution
  reduction} with respect to $P$ is $\hookrightarrow^1_P \circ
\rightarrow^{\mu}_P$.
\end{definition}

It is not hard to see that the reduction relation $\leadsto_P$ models
traditional SLD-resolution steps~\cite{Llo87} with respect to $P$,
and, writing $\leadsto_{s\,P}$ for $\hookrightarrow^1_P \circ
\rightarrow^{\mu}_P$, that the reduction relation $\leadsto_{s\,P}$
models S-resolution steps with respect to $P$.  If an SLD-resolution,
rewriting, or S-resolution reduction with respect to $P$ starts with
$[t]$, then we say it is a reduction {\em for $t$} with respect to
$P$.  If there exists an $n$ such that $[t] \leadsto^n_P [\;]$ or $[t]
\leadsto^n_{s\,P} [\;]$, then we say that this reduction for $t$ is
\emph{inductively successful}. For SLD-resolution reductions this
agrees with standard logic programming terminology.

If we regard the term $t$ as a ``query", then we may regard the
composition $\sigma_n \circ \ldots \circ \sigma_1$ of the
substitutions $\sigma_1, \ldots , \sigma_n \in \cat{Subst}(\Sigma)$
involved in the steps of an inductively successful SLD-resolution
reduction for $t$ as an ``answer'' to this query, and we may think of
the reduction as computing this answer.  Such a composition for an
initial sequence of SLD-resolution reductions in a possibly
non-terminating SLD-resolution reduction for $t$ can similarly be
regarded as computing a partial answer to that query. We use this
terminology for rewriting and S-resolution reductions as well.

\begin{example}
The following are SLD-resolution, rewriting, and S-resolution
reductions, respectively, with respect to $P_2$:
\begin{itemize}
\item $[\mathtt{nats(X)}] \leadsto [\mathtt{nat(X')},
  \mathtt{nats(Y)}] \leadsto [\mathtt{nats(Y)}] \leadsto
  [\mathtt{nat(X'')}, \mathtt{nats(Y')}]\leadsto \ldots$

\vspace*{0.1in}

\item $[\mathtt{nats(X)}]$

\vspace*{0.1in}

\item $[\mathtt{nats(X)}] \rightarrow^{\mu} [\mathtt{nats(X)}]
  \hookrightarrow^1 [\mathtt{nats(scons(X',Y))}] \rightarrow^{\mu}
                 [\mathtt{nat(X')}, \mathtt{nats(Y)}]
                 \hookrightarrow^1 [\mathtt{nat(0)}, \mathtt{nats(Y)}]
                 \rightarrow^{\mu} [\mathtt{nats(Y)}]
                 \hookrightarrow^1 [\mathtt{nats(scons(X'',Y'))}]
                 \rightarrow^{\mu} \ldots$
\end{itemize}
\noindent
In the S-resolution reduction above, $[\mathtt{nats(X)}]
\rightarrow^{\mu} [\mathtt{nats(X)}]$ in 0 steps, since
$[\mathtt{nats(X)}]$ is already in $\rightarrow$-normal form. The
initial sequences of the SLD-resolution and S-resolution reductions
each compute the partial answer $\{\mathtt{X} \mapsto
\mathtt{scons(0,scons(X'', Y'))}\}$ to the query $\mathtt{nats(X)}$.
\end{example}

The observation that, even for coinductive program like $P_2$,
$\rightarrow^{\mu}$ reductions are finite and thus can serve as
measures of finite observation, has led to the following definition of
observational productivity in LP, first introduced in~\cite{KPS12-2}:

\begin{definition}\label{def:prod}
A program $P \in \cat{LP}(\Sigma)$ is {\em observationally productive}
if $\rightarrow_P$ is strongly normalising, i.e., if every rewriting
reduction with respect to $P$ is finite.
\end{definition}

\begin{example}
The programs $P_1, P_2, P_3, P_4,$ and $P_5$ are all observationally
productive, as is the program $P_7$ defined in
Example~\ref{ex:overlap} below. By contrast, the ``bad" program of
Example~\ref{ex:bad} and the program $P_6$ defined in
Example~\ref{ex:conn} below are not.
\end{example}

\noindent
A similar notion of observational productivity, in terms of strong
normalisation of term rewriting, has recently been introduced for
copatterns in functional programming~\cite{Basold2015}. 

A general analysis of observational productivity for LP is rather
subtle. Indeed, there are programs $P$ and queries $t$ for which there
are inductively successful SLD-resolution reductions, but for which
$P$ nevertheless fails to be observationally productive because there
exist no inductively successful S-resolution reductions.

\begin{example}\label{ex:conn}
Consider the graph connectivity program~\cite{SS86} $P_6$ given by:

\vspace*{0.05in}\noindent
0. $\mathtt{conn(X,Y)} \gets \mathtt{conn(X,Z)}, \mathtt{conn(Z,Y)}$\\
1. $\mathtt{conn(a,b)} \gets$\\
2. $\mathtt{conn(b,c)} \gets$

\vspace*{0.05in}\noindent Although there exist inductively successful
SLD-resolution reductions for $\mathtt{conn(X,Y)}$ with respect to
$P_6$, there are no such inductively successful S-resolution
reductions. Indeed, the only S-resolution reductions for
$\mathtt{conn(X,Y)}$ with respect to $P_6$ are infinite rewriting
reductions that, with each rewriting reduction, accumulate an
additional term involving $\mathtt{conn}$. A representative example of
such an S-resolution reduction is
\[ [\mathtt{conn(X,Y)}] \rightarrow [\mathtt{conn(X,X')},
   \mathtt{conn(X',Y)}] \rightarrow [\mathtt{conn(X,X'')},
  \mathtt{conn(X'',X')}, \mathtt{conn(X',Y)}] \rightarrow \ldots\]
\noindent
Thus, $P_6$ is not observationally productive.
\end{example}  

\vspace*{0.05in}

With this in mind, we first turn our attention to analysing the
inductive properties of S-resolution reductions.

\subsection{Inductive Properties of S-Resolution Reductions}\label{sec:semantics}

In this section, we discuss whether, and under which conditions,
S-resolution reductions are inductively sound and complete. First we
recall that SLD-resolution is inductively sound and
complete~\cite{Llo87}. The standard results of inductive soundness and
completeness for SLD-resolution~\cite{Llo87} can be summarised as:

\begin{theorem}\label{thm:sld}
Let $P \in \cat{LP}(\Sigma)$ and $t \in \cat{Term}(\Sigma)$.
\begin{itemize}
\item (Inductive soundness of SLD-resolution reductions) If $t
  \leadsto_P^n []$ for some $n$ and computes answer $\theta$, then
  there exists a term $t' \in \mathbf{GTerm}(\Sigma)$ such that $ t'
  \in M_P$ and $t'$ is an instance of $\theta(t)$.
\item (Inductive completeness of SLD-resolution reductions) If $t \in
  M_P$, then there exists a term $t' \in \mathbf{Term}(\Sigma)$ that
  yields an SLD-resolution reduction $t' \leadsto_P^n []$ that
  computes answer $\theta \in \cat{Subst}(\Sigma)$ such that $t$ is an
  instance of $\theta(t')$.
\end{itemize}
\end{theorem}

We now show that, in contrast to SLD-resolution reductions,
S-resolution reductions are inductively sound but incomplete. We first
establish inductive soundness.

\begin{theorem}\label{thm:cs}(Inductive soundness of S-resolution reductions) 
If $t \ \leadsto_{s\,P}^n \ []$ for some $n$ and computes answer
$\theta$, then there exists a term $t' \in \mathbf{Term}(\Sigma)$ such
that $t' \in M_P$ and $t'$ is an instance of $\theta(t)$.
\end{theorem}

\begin{proof}
The proof is by induction on $n$ in $\leadsto_{s\,P}^n$. It is a
simple adaptation of the soundness proof for SLD-resolution reductions
given in, e.g.,~\cite{Llo87}.
\end{proof}

To show that S-resolution reductions are not inductively complete, it
suffices to provide one example of a program $P$ and a term $t$ such
that $P \models \theta(t)$ but no inductively successful S-resolution
reduction exists for $t$. We will in fact give two such examples, each
of which is representative of a different way in which S-resolution
reductions can fail to be inductively complete.

\begin{example}\label{ex:incomplete}
Consider $P_6$ and the S-resolution reduction shown in
Example~\ref{ex:conn}.  The instantiation $\mathtt{conn(a,c)}$ of
$\mathtt{conn(X,Y)}$ is in the least Herbrand model of $P_6$, but
there are no finite S-resolution reductions, and therefore no
inductively successful S-resolution reductions, for $P_6$ and the
query $\mathtt{conn(X,Y)}$. This shows that programs that are not
observationally productive need not be inductively complete.
\end{example}

In light of Example~\ref{ex:incomplete} it is tempting to try to prove
the inductive completeness of S-resolution reductions for
observationally productive logic programs only. However, this would
not solve the problem, as the following example confirms:

\begin{example}\label{ex:overlap}
Consider the program $P_7$ given by:

\vspace*{0.05in}\noindent
0. $\mathtt{p(c)} \gets$\\
1. $\mathtt{p(X)} \gets \mathtt{q(X)}$

\vspace*{0.05in}\noindent We have that $P_7 \models \mathtt{p(c)}$ for
the instantiation $\mathtt{p(c)}$ of $\mathtt{p(X)}$, but there is no
inductively successful S-resolution reduction for $P_7$ and
$\mathtt{p(X)}$.
\end{example}

Program $P_7$ is an example of \emph{overlapping} program, i.e., a
program containing clauses whose heads unify. We could show that, for
programs that are both observationally productive and non-overlapping,
S-resolution reductions are inductively complete. However, restricting
attention to non-overlapping programs would seriously affect
generality of our results, and would have the effect of making
S-resolution even less suited for inductive proof search than
SLD-resolution. We prefer instead to refine S-resolution so that it is
inductively complete for {\em all} programs. The question is whether
or not such refinement is possible.

An intuitive answer to this question comes from reconsidering
Example~\ref{ex:overlap}. There, the interleaving of
$\rightarrow^{\mu}$ and $\hookrightarrow^1$ has the effect of
restricting the search space. Indeed, once the rewriting portion of
the only possible S-resolution reduction on $\mathtt{p(X)}$ is
performed, the new subgoal $\mathtt{q(X)}$ prevents us from revisiting
the initial goal $\mathtt{p(X)}$ and unifying it with the clause
$P_7(0)$, as would be needed for an inductively successful
S-resolution reduction. This is how we lose inductive completeness of
the proof search.

One simple remedy would be to redefine S-resolution reductions to be
$(\hookrightarrow^1 \circ \rightarrow^n)$-reductions, where $n$ ranges
over all non-negative integers. This would indeed restore inductive
completeness of S-resolution for overlapping programs. But it would at
the same time destroy our notion of observational productivity, which
depends crucially on $\rightarrow^{\mu}$. An alternative solution
would keep our definitions of S-resolution reductions and
observational productivity intact, but also find a way to keep track
of all of the unification opportunities arising in the proof search.
This is exactly the route we take here.

Kowalski~\cite{Kow74} famously observed that \emph{Logic Programming =
  Logic + Control}. For Kowalski, the logic component was given by
SLD-resolution reductions, and the control component by an algorithm
coding the choice of search strategy. As it turns out, SLD-resolution
reductions are sound and complete irrespective of the control
component.  What we would like to do in this paper is revise the very
logic of LP by replacing SLD-resolution reductions with S-resolution
reductions. 
As it turns out, defining a
notion of S-resolution that is both inductively complete and capable
of capturing observational productivity requires imposing an
appropriate notion of control on this logic. In the next section we
therefore define S-resolution in terms of rewriting trees.
 Rewriting trees allow us to neatly integrate
precisely the control on S-resolution reductions needed to achieve
both of these aims for the underlying logic of S-resolution
reductions. We thus arrive at our own variant of Kowalski's formula,
namely \emph{Structural Logic Programming = Logic + Control} --- but
now the logic is given by S-resolution reductions and the control
component is captured by rewriting trees. The remainder of the paper
is devoted to developing the above formula into a formal theory.

\section{Inductive Soundness and Completeness of
  Structural Resolution}\label{sec:ST}
	
To ensure that S-resolution reductions are inductively complete, we
need to impose more control on the rewriting reductions involved in
them. To do this, we first note that the rewriting reduction in
Example~\ref{ex:overlap} can be represented as the tree

\vspace*{-0.15in}

\begin{center}
\begin{tikzpicture}[level 1/.style={sibling distance=15mm},
level 2/.style={sibling distance=15mm}, level 3/.style={sibling
distance=15mm},scale=.8,font=\footnotesize,baseline=(current bounding
box.north),grow=down,level distance=9mm]
\node (root) {$\mathtt{p(X)}$}
	child { node {$\mathtt{q(X)}$}};
 \end{tikzpicture}
\end{center}

Now, we would also like to reflect within this tree the fact that
$\mathtt{p(X)}$ to unifies with the head of clause $P_7(0)$ and, more
generally, to reflect the fact that any term can, in principle, unify
with the head of any clause in the program. We can record these
possible unifications in tree form, as follows:

\vspace*{-0.1in}

\begin{center}
\begin{tikzpicture}[level 1/.style={sibling distance=15mm},
level 2/.style={sibling distance=15mm}, level 3/.style={sibling
distance=15mm},scale=.8,font=\footnotesize,baseline=(current bounding
box.north),grow=down,level distance=9mm]
\node (root) {$\mathtt{p(X)}$}
	child { node {$P_7(0)?$}}
  	child { node {$\mathtt{q(X)}$}
		child { node {$P_7(0)?$}}
        	child { node {$P_7(1)?$}}};
  \end{tikzpicture}
\end{center}

We can now follow-up each of these possibilities and in this way
extend our proof search. To do this formally, we distinguish two kinds
of nodes: {\em and-nodes}, which capture terms coming from clause
bodies, and {\em or-nodes}, which capture the idea that every term
can, in principle, match several clause heads. We also introduce {\em
  or-node variables} to signify the possibility of unifying a term
with the head of a clause when the matching of that term against that
clause head fails. This careful tracking of possibilities allows us to
construct the inductively successful S-resolution reduction for
$\mathtt{p(X)}$ and program $P_7$ shown in Figure~\ref{fig:overlap}.
The figure depicts two rewriting trees, each modelling all possible
rewriting reductions for the given query (represented as a goal
clause) with respect to $P_7$. Rewriting trees have alternating levels
of or-nodes and and-nodes, as well as or-node variables ($X_1$, $X_2$,
and $X_3$ in the figure) ranging over rewriting trees. By unifying
$\mathtt{p(X)}$ with $P_7(0)$ we replace the or-node variable $X_1$ in
the first rewriting tree with a new rewriting tree (in this case
consisting of just the single node $\mathtt{p(c) \gets}$) to
transition to the second rewriting tree shown. When a node contains a
clause, such as $\mathtt{p(c) \gets}$, that has an empty body, it is
equivalent to an empty subgoal. Thus, the underlined subtree of the
second rewriting tree in Figure~\ref{fig:overlap} represents the
inductively successful S-resolution reduction $P_7 \vdash
\mathtt{p(c)} \rightarrow []$.

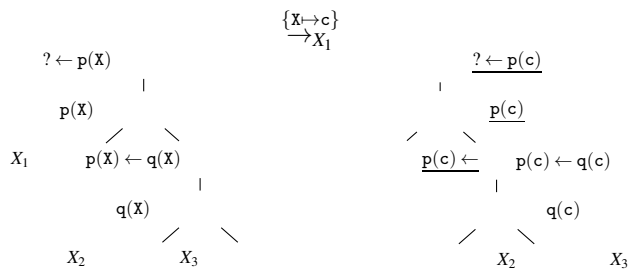
\begin{figure*}
\begin{center}
\begin{tikzpicture}[level 1/.style={sibling distance=15mm},
level 2/.style={sibling distance=25mm},
level 3/.style={sibling distance=25mm},scale=.6,font=\scriptsize,baseline=(current bounding box.north),grow=down,level distance=11mm]
\hspace*{-0.35in}\node (root) {$? \gets   \mathtt{p(X)}$}
	child { node {$\mathtt{p(X)}$} child { node {$X_1$}} child {
                   node {$\mathtt{p(X) \gets q(X)}$} child {node{
                   $\mathtt{q(X)}$} child { node {$X_2$}} child { node
                   {$X_3$}} } } }; \end{tikzpicture}
                   $\stackrel{\{\mathtt{X} \mapsto \mathtt{c}\}}{\rightarrow_{X_1}}$
\begin{tikzpicture}[level 1/.style={sibling distance=25mm},
level 2/.style={sibling distance=25mm},
level 3/.style={sibling distance=25mm},scale=.6,font=\scriptsize,baseline=(current bounding box.north),grow=down,level distance=11mm]
\hspace*{0.35in} \node (root) {\underline{$? \gets   \mathtt{p(c)}$}}
	child { node {\underline{$\mathtt{p(c)}$}} child { node
                   {\underline{$\mathtt{p(c) \gets}$}}} child { node
                   {$\mathtt{p(c) \gets q(c)}$} child {node{
                   $\mathtt{q(c)}$} child { node {$X_2$}} child { node
                   {$X_3$}} } }
                   }; \end{tikzpicture}\hspace*{-0.2in} \end{center} 
\caption{\footnotesize{A tree transition for the overlapping program
    $P_7$ and the goal clause $\gets \mathtt{p(X)}$. Underlined in the
    second rewriting tree is a inductively successful S-resolution
    reduction.}} \vspace*{-0.1in}\label{fig:overlap}\vspace*{0.2in}
\end{figure*}

\subsection{Modeling $\rightarrow^{\mu}$ by Rewriting Trees}

We now proceed to define the construction formally. For this, we first
observe that a clause $C$ over a signature $\Sigma$ that is of the
form $A \gets B_0, \ldots , B_n$ can be naturally represented as the
total function (also called $C$) from the finite tree language $L =
\{\epsilon,0,...,n\}$ of depth $1$ to $\mathbf{Term}(\Sigma)$ such
that $C(\epsilon) = A$ and $C(i) = B_i$ for $i = \mathit{dom}(C)
\setminus \{\epsilon\}$. With this representation of clauses in hand,
we can formalise our notion of a rewriting tree.

\begin{definition}\label{def:CT}
Let $V_R$ be a countably infinite set of variables disjoint from
$\mathit{Var}$. If $P\in \cat{LP}(\Sigma)$, $C\in
\cat{Clause}(\Sigma)$, and $\sigma \in \mathbf{Subst}(\Sigma)$ is
idempotent, then the tree $\kw{rew}(P,C, \sigma)$ is the function
$T:\mathit{dom}(T) \rightarrow \cat{Term}(\Sigma) \cup
\cat{Clause}(\Sigma) \cup V_R$, where $\mathit{dom}(T) \neq \emptyset$
is a tree language defined simultanously with $\kw{rew}(P,C,\sigma)$,
such that:
\begin{enumerate}
\item $T(\epsilon) = \sigma(C)$ and, for all $i \in \mathit{dom}(C)
  \setminus \{\epsilon\} $, $T(i) = \sigma(C(i))$.
\item For $w\in\mathit{dom}(T)$ with $|w|$ even and $|w| > 0$, $T(w)
  \in \cat{Clause}(\Sigma) \cup V_R$. Moreover,
\begin{enumerate}
\item If $T(w) \in V_R$, then $\{j\mid wj \in \mathit{dom}(T)\} =
  \emptyset$.
\item If $T(w) = B \in \cat{Clause}(\Sigma)$, then there exists a
  clause $P(i)$ and a $\theta \in \cat{Subst}(\Sigma)$ such that $\mathit{head}(P(i)
  \prec_\theta \mathit{head}(B)$.
Moreover, for every $j \in \mathit{dom}(P(i))
  \setminus \{\epsilon\}$, $wj\in \mathit{dom}(T)$ and $T(wj) =
  \sigma(\theta(P(i)(j)))$.
\end{enumerate}
\item For $w\in\mathit{dom}(T)$ with $|w|$ odd, $T(w) \in
  \cat{Term}(\Sigma)$. Moreover, for every $i\in \mathit{dom}(P)$, we
  have 
\begin{enumerate}
\item[(a)] $wi\in \mathit{dom}(T)$.
\item[(b)] $T(wi) =
    \begin{cases}
 \sigma(\theta(P(i))) & \text{if } \mathit{head}(P(i)) \prec_\theta T(w) \\
\text{a fresh } X\in V_R
 &\text{otherwise}
    \end{cases}$
\end{enumerate}
\item No other words are in $\mathit{dom}(T)$.
\end{enumerate}
$T(w)$ is an {\em or-node} if $|w|$ is even, and an {\em and-node} if
$|w|$ is odd.

We assume, as is standard in LP, that all variables are \emph{renamed apart} when terms are matched or unified against the program clauses.

\end{definition}

If $P \in \cat{LP}(\Sigma)$, then $T$ is a {\em rewriting tree for
  $P$} if it is either the empty tree or $\kw{rew}(P,C,\sigma)$ for
some $C$ and $\sigma$. Since mgms are unique up to variable renaming,
$\kw{rew}(P,C, \sigma)$ is as well.  A rewriting tree for a program
$P$ is finite or infinite according as its domain is finite or
infinite.  We write $\mathbf{Rew}(P)$ ($\mathbf{Rew}^\infty(P)$,
$\mathbf{Rew}^\omega(P)$) for the set of all finite (infinite, {\em
  all}) rewriting trees for $P$.

This style of tree definition mimics the classical style of defining
terms as maps from a tree language to a given
domain~\cite{Llo87,Courcelle83}. As with tree representations of
terms, {\em arity constraints} are imposed on rewriting trees. The
arity constraints in items 2b and 3a specify that the arity of an
and-node is the number of clauses in the program and the arity of an
or-node is the number of terms in its clause body. The arity
constraint in item 2a specifies that or-node variables must have arity
$0$. Or-node variables indicate where in a rewriting tree substitution
can take place.

\begin{example}
The rewriting trees $\kw{rew}(P_3,? \gets \mathtt{fibs(0,s(0),X)},
\mathit{id})$, $\kw{rew}(P_3,$ $? \gets
\mathtt{fibs(0,s(0),cons(0,S))}, \mathit{id})$, and $\kw{rew}(P_3,?
\gets \mathtt{fibs(0,s(0),cons(0,S))},$ $\mathtt{\{Z
  \mapsto s(0)\}})$ are shown in Figure~\ref{fig:fibs2}.  Note the
or-node variables and the arities.  An or-node can have arity $0$,
$1$, or $2$ according as its clause body contains $0$, $1$, or $2$
terms, and every and-node has arity $3$ because $P_3$ has three
clauses.
\end{example}

Although perhaps mysterious at first, the third parameter $\sigma$ in
Definition~\ref{def:CT} for $T = \kw{rew}(P,C,\sigma)$ is necessary
account for variables occurring in $T$ not affected by mgms computed
during $T$'s construction. It plays a crucial role in ensuring that
applying a substitution to a rewriting tree again yields a rewriting
tree~\cite{JohannKK15}. We have:

\begin{definition}\label{def:substs}
Let $P \in \textbf{LP}(\Sigma)$, $C \in \mathbf{Clause}(\Sigma)$,
$\sigma,\sigma' \in \mathbf{Subst}(\Sigma)$ idempotent, and $T =
\kw{rew}(P,C,\sigma)$. Then the rewriting tree $\sigma'(T)$ is defined
as follows:
\begin{itemize}
\item for every $w \in \mathit{dom}(T)$ such that $T(w)$ is an
  and-node or non-variable or-node, $(\sigma'(T))(w) = \sigma'(T(w))$.
\item for every $wi \in \mathit{dom}(T)$ such that $T(wi) \in V_R$, if
  $\mathit{head}(P(i)) \prec_\theta \sigma'(T)(w)$, then
  $(\sigma'(T))(wiv) = \kw{rew}(P,\theta(P(i)),\sigma'\sigma)(v)$.
  (Note that $v = \epsilon$ is possible.)  If no mgm of
  $\mathit{head}(P(i))$ against $\sigma'(T)(w)$ exists, then
  $(\sigma'(T))(wi) = T(wi)$.
\end{itemize}
\end{definition}

\noindent
Both conditions in the above definition are critical for ensuring that
$\sigma'(T)$ satisfies Definition~\ref{def:CT}. We then have the
following substitution theorem for rewriting trees. It is proved
in~\cite{JohannKK15}. The operation of substitution on rewriting trees is also introduced  in the same way in~\cite{BonchiZ15}.

\begin{theorem}\label{prop:sub-props}
Let $P \in \textbf{LP}(\Sigma)$, $C \in \cat{Clause}(\Sigma)$, and
$\theta, \sigma \in \mathbf{Subst}(\Sigma)$. Then $\theta(\kw{rew}
(P, C, \sigma)) = \kw{rew} (P, C, \theta\sigma)$.
\end{theorem}

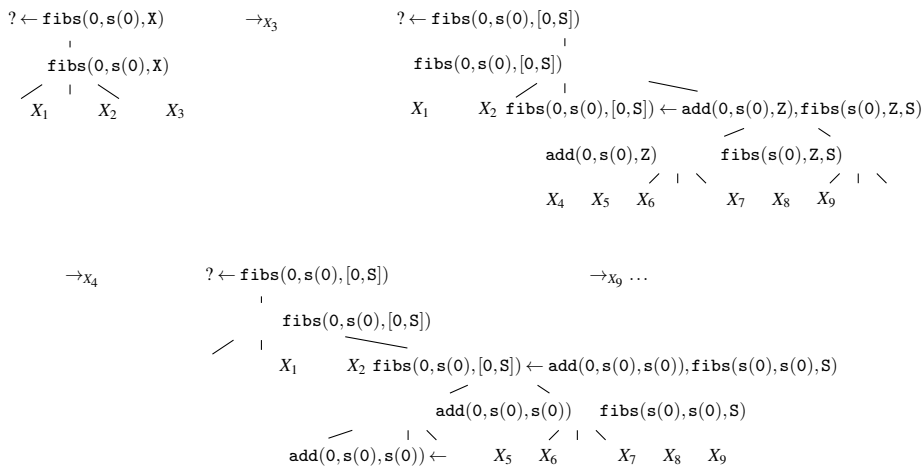
\begin{figure*}
\begin{tikzpicture}[level 1/.style={sibling distance=15mm},
level 2/.style={sibling distance=15mm}, level 3/.style={sibling
  distance=15mm},scale=.6,font=\scriptsize,baseline=(current bounding
box.north),grow=down,level distance=10mm]
\hspace*{0.2in}
\node (root) {$? \gets
  \mathtt{fibs(0,s(0),X)} \hspace*{0.4in}  \rightarrow_{X_3}$  \hspace*{-0.4in}}
 	child { node {$\mathtt{fibs(0,s(0),X)}$}
                   child { node {$X_1$}}
         		child { node {$X_2$}}
			child { node {$X_3$}}
};
  \end{tikzpicture}\hspace*{-0.4in}\vspace*{0.2in}	
\begin{tikzpicture}[level 1/.style={sibling distance=15mm},
level 2/.style={sibling distance=15mm},
level 3/.style={sibling distance=15mm},
level 4/.style={sibling distance=10mm},
level 5/.style={sibling distance=10mm},scale=.6,font=\scriptsize,baseline=(current bounding box.north),grow=down,level distance=10mm]
\hspace*{-0.4in}
  \node (root) {$? \gets
  \mathtt{fibs(0,s(0),[0,S])}$}
	child { node {$\mathtt{fibs(0,s(0),[0,S])}$}
                   child { node {$X_1$}}
									child { node {$X_2$}}
									child [sibling distance=50mm] { node {$\mathtt{fibs(0,s(0),[0,S])} \gets \mathtt{add(0,s(0),Z)}, \mathtt{fibs(s(0),Z,S)}$}
									   child  [sibling distance=50mm] { node {$\mathtt{add(0,s(0),Z)}$}
										  child { node {$X_4$}}
									child { node {$X_5$}}
									child { node {$X_6$}}}
									child  [sibling distance=30mm]  { node {$\mathtt{fibs(s(0),Z,S)}$}
									  child { node {$X_7$}}
									child { node {$X_8$}}
									child { node {$X_9$}}}
									}
									};
 \end{tikzpicture}
\begin{tikzpicture}[level 1/.style={sibling distance=15mm},
level 2/.style={sibling distance=15mm},
level 3/.style={sibling distance=15mm},
level 4/.style={sibling distance=10mm},
level 5/.style={sibling distance=10mm},scale=.6,font=\scriptsize,baseline=(current bounding box.north),grow=down,level distance=10mm]
\hspace*{0.5in}  \node (root) {$\rightarrow_{X_4}$ \hspace*{0.5in} $? \gets
    \mathtt{fibs(0,s(0),[0,S])}$ \hspace*{1in} $\rightarrow_{X_9} \ldots$}
	child { node {$\mathtt{fibs(0,s(0),[0,S])}$}
                   child { node {$X_1$}}
			child { node {$X_2$}}
			child [sibling distance=55mm] { node {$\mathtt{fibs(0,s(0),[0,S])} \gets \mathtt{add(0,s(0),s(0))}, \mathtt{fibs(s(0),s(0),S)}$}
         		   child  [sibling distance=45mm] { node
                             {$\mathtt{add(0,s(0),s(0))}$}
										  child  [sibling distance=30mm]  { node {$\mathtt{add(0,s(0),s(0))} \gets$ }}
									child { node {$X_5$}}
									child { node {$X_6$}}}
									child  [sibling distance=30mm]  { node {$\mathtt{fibs(s(0),s(0),S)}$}
									  child { node {$X_7$}}
									child { node {$X_8$}}
									child { node {$X_9$}}}
									}
									};
\end{tikzpicture}
\caption{\footnotesize{An initial fragment of an S-derivation for
    $P_3$ and $\mathtt{fibs(0,s(0),X)}$. The three rewriting
    trees shown are $\kw{rew}(P_3,? \gets \mathtt{fibs(0,s(0), X)},
    id)$, $\kw{rew}(P_3,$ $? \gets \mathtt{fibs(0,s(0),cons(0,S))},
    id)$, and $\kw{rew}(P_3,? \gets \mathtt{fibs(0,s(0),cons(0,S))},$
    $\mathtt{\{Z \mapsto s(0)\}})$, respectively.  To save space in
    the figure we abbreviate $\mathtt{cons}(\_,\_)$ by $[\_,\_]$, and
    similarly below.}}\label{fig:fibs2}
\end{figure*}

We can now formally establish the relation between rewriting
reductions and rewriting trees. We first have the following
proposition, which is an immediate consequence of
Definitions~\ref{def:prod} and~\ref{def:CT}:

\begin{proposition}\label{prop:RT-prod}
$P\in \cat{LP}(\Sigma)$ is observationally productive iff, for every
  term $t \in \cat{Term}(\Sigma)$ and every substitution $\sigma \in
  \cat{Subst}(\Sigma)$, $\kw{rew}(P,?\gets t, \sigma)$ is finite.
\end{proposition}

We can further establish a correspondence between certain subtrees of
rewriting trees and inductively successful S-resolution reductions.

\begin{definition}
A tree $T'$ is a \emph{rewriting subtree} of a rewriting tree $T$ if
$\mathit{dom}(T') \subseteq \mathit{dom}(T)$ and the following
properties hold:
\begin{enumerate}
\item $T'(\epsilon) = T(\epsilon)$.
\item If $w \in 
dom(T')$ with $|w|$ even, then $T'(w) = T(w)$,
$wi \in dom(T')$ for every $wi \in dom(T)$, and $T'(wi) = T(wi)$.
\item If $w \in 
dom(T')$ with $|w|$ odd, then $T'(w) = T(w)$, there exists a unique
$i$ with $wi \in dom(T)$ such that $wi \in dom(T')$, and $T'(wi) =
T(wi)$ for this $i$.
\end{enumerate}
\end{definition}   

\noindent
Rewriting subtrees can be either finite or infinite. Note that the
and-nodes in item 2 grow children by universal quantification, whereas
the or-nodes in item 3 grow them by existential quantification.

\begin{definition}\label{def:success/failure}
If $T \in \textbf{Rew}^\omega(P)$, then an or-node of $T$ is an {\em
  inductive success node} if it is a non-variable leaf node of $T$. If
$T'$ is a finite rewriting subtree of $T$ all of whose leaf nodes are
inductive success nodes of $T$, then $T'$ is an {\em inductive success
  subtree} of $T$. If $T$ contains an inductive success subtree then
we call $T$ an {\em inductive success tree}.
\end{definition}

The following proposition is immediate from Definitions~\ref{def:CT}
and~\ref{def:success/failure}:

\begin{proposition}\label{prop:RT-red}
If $P\in \cat{LP}(\Sigma)$ and $t\in \cat{Term}(\Sigma)$, then $P
\vdash t \rightarrow^n [\;]$ for some $n$ iff $\kw{rew}(P,?\gets t,
\mathit{id})$ is an inductive success tree.
\end{proposition}

With these preliminary results in hand we can now begin to show that
rewriting trees impose on S-resolution reductions precisely the
control required to prove their soundness and completeness with
respect to least Herbrand models. We first observe that:

\begin{theorem}\label{th:isc} 
Let $P \in \cat{LP}(\Sigma)$ and $t \in \mathbf{Term}(\Sigma)$.
\begin{itemize}
\item If $\kw{rew}(P,? \gets t, \sigma)$ is an inductive success tree
  for some $\sigma \in \cat{Subst}(\Sigma)$ then, for every instance
  $t' \in \cat{GTerm}(\Sigma)$ of $\sigma(t)$, $t' \in M_P$.
\item If $ t \in M_P$, then there exists a grounding substitution
  $\theta \in \cat{Subst}(\Sigma)$ such that $\kw{rew}(P,? \gets t,
  \theta)$ is an inductive success tree.
\end{itemize}
\end{theorem}
\begin{proof}
The proof is by induction on the depth of rewriting trees.
\end{proof}

\begin{example}
The term $\mathtt{conn(a,c)}$ is in $M_{P_6}$. The tree
$\kw{rew}(P_6,?\gets \mathtt{conn(a,c)}, \mathit{id})$ is not an
inductive success tree, as Figure~\ref{fig:conn} shows. However,
$\kw{rew}(P_6,?\gets \mathtt{conn(a,c)}, \theta)$, for $\theta =
\{\mathtt{Z}\mapsto \mathtt{b}\}$, is indeed an inductive success
tree. This accords with Theorem~\ref{th:isc}.
\end{example}

\begin{figure*}
\begin{center}
 \begin{tikzpicture}[level distance=10mm,sibling
     distance=50mm,scale=.6,font=\scriptsize,baseline=(current
     bounding box.north),grow=down] 
\hspace*{-0.3in} \node {? $\gets \mathtt{conn(a,c)} \hspace*{0.5in}
  \rightarrow_{X_3}
    \hspace*{-0.5in}$}
		child[sibling distance=20mm]{node
                  {$\mathtt{conn(a,c)}$} 
		child[sibling distance=40mm] {node
                  {$\mathtt{conn(a,c)} \gets
                    \mathtt{conn(a,Z),conn(Z,c)}$} 
	child[sibling distance=30mm]{node {$\mathtt{conn(a,Z)}$}
			child[sibling distance=10mm]{node {$\vdots$}}
				child[sibling distance=10mm]{node {$X_3$}}	
			   child[sibling distance=10mm]{node {$X_4$}}
			}
	child[sibling distance=30mm]{node {$\mathtt{conn(Z,c)}$}
				child[sibling distance=10mm]{node {$\vdots$}}	
			   child[sibling distance=10mm]{node {$X_5$}}
	child[sibling distance=10mm]{node {$X_{6}$}}
			}}
		child[sibling distance=20mm]{node {$X_1$}}
		child[sibling distance=10mm]{node {$X_2$}}};  
  \end{tikzpicture}
	 \begin{tikzpicture}[level distance=10mm,sibling distance=50mm,scale=.6,font=\scriptsize,baseline=(current bounding box.north),grow=down ]
  \node {\underline{? $\gets \mathtt{conn(a,c)}$}}
		child[sibling distance=20mm]{node
                  {\underline{$\mathtt{conn(a,c)}$}} 
		child[sibling distance=40mm] {node
                  {\underline{$\mathtt{conn(a,c)} \gets
                    \mathtt{conn(a,b),conn(b,c)}$}} 
	child[sibling distance=40mm]{node {\underline{$\mathtt{conn(a,b)}$}}
					child[sibling distance=20mm]{node {$\vdots$}}	
			   child[sibling distance=20mm]{node {\underline{$\mathtt{conn(a,b)} \gets$}}}
	child[sibling distance=20mm]{node {$X_4$}}
			}
	child[sibling distance=40mm]{node {\underline{$\mathtt{conn(b,c)}$}}
				child[sibling distance=10mm]{node {$\vdots$}}	
			   child[sibling distance=10mm]{node {$X_5$}}
	child[sibling distance=20mm]{node {\underline{$\mathtt{conn(b,c)} \gets$}}}
			}}
		child[sibling distance=20mm]{node {$X_1$}}
		child[sibling distance=10mm]{node {$X_2$}}};  
  \end{tikzpicture}
\end{center}
\caption{\footnotesize{An S-refutation for the program $P_6$ of
    Example~\ref{ex:conn} and $\mathtt{conn(a,c)}$. The left tree
    $\kw{rew}(P_6,?\gets \mathtt{conn(a,c)}, \mathit{id})$ is not an
    inductive success tree. However, the right tree
    $\kw{rew}(P_6,?\gets \mathtt{conn(a,c)},
    \{\mathtt{Z}\mapsto\mathtt{b}\})$ is. The inductive success
    subtree of the right tree is underlined. }}
\label{fig:conn}
\end{figure*}

\subsection{Modeling $\hookrightarrow$ by Transitions
Between Rewriting Trees} 

Next we define transitions between rewriting trees. Such transitions
are defined by the familiar notion of a resolvent, and assume a
suitable algorithm for renaming ``free" clause variables
apart~\cite{JohannKK15}. Let $P \in \textbf{LP}(\Sigma)$ and $t\in
\cat{Term}(\Sigma)$. If $\mathit{head}(P(i)) \sim_\theta t$, then
$\theta$ is called the \emph{resolvent} of $P(i)$ and $t$. If no such
$\theta$ exists then $P(i)$ and $t$ have {\em null resolvent}. A
non-null resolvent is an \emph{internal resolvent} if it $\mathit{head}(P(i))
\prec_{\theta} t$ and an {\em external resolvent} otherwise.

\begin{definition}\label{def:resapp}
 Let $T = \kw{rew}(P,C,\sigma) \in \textbf{Rew}^\omega(P)$. If $X =
 T(wi) \in V_R$, then the rewriting tree $T_{X}$ is defined as
 follows: If the external resolvent $\theta$ for $P(i)$ and $T(w)$ is
 null, then $T_X$ is the empty tree. If $\theta$ is non-null, then
 $T_X = \kw{rew}(P, C, \theta \sigma)$.
\end{definition}

If $X \in V_R$, we denote the computation of $T_X$ from $T \in
\mathbf{Rew}^\omega(\Sigma)$ by $T \rightarrow T_X$. The operation $T
\rightarrow T_X$ is a \emph{tree transition} for $P$ and $C$;
specifically, we call the tree transition $T \rightarrow T_X$ the tree
transition for $T$ with respect to $X$.  A {\em tree transition} for
$P \in \mathbf{LP}(\Sigma)$ is a tree transition for $P$ and some $C
\in \mathbf{Clause}(\Sigma)$.  If $T \rightarrow T_X$ is a tree
transition and if $X = T(w)$, then we say that both the node $T(w)$
and the branch of $T$ that this node lies on are {\em expanded} in
this transition. A (finite or infinite) sequence $T_0 = \kw{rew}(P,?
\gets t,\mathit{id}) \rightarrow T_1 \rightarrow T_2 \rightarrow
\ldots$ of tree transitions for $P$ is a {\em structural tree
  resolution derivation}, or simply an \emph{S-derivation} for short,
for $P$ and $t$. An S-derivation for $P$ and $t$ is said to be an {\em
  S-refutation}, or an {\em inductive proof}, for $t$ with respect to
$P$, if it is of the form $T_0 \rightarrow T_1 \rightarrow
... \rightarrow T_n$ for some $n$, where $T_n$ is an inductive success
tree.  Figure~\ref{fig:fibs2} shows an initial fragment of an infinite
S-derivation for the program $P_3$ and $\mathtt{fibs(0,s(0),X)}$.  The
derivations shown in Figures~\ref{fig:overlap} and \ref{fig:conn} are
inductive proofs for $P_7$ and $\mathtt{p(c)}$, and for $P_6$ and
$\mathtt{conn(a,c)}$, respectively. Note that the final trees of
Figures~\ref{fig:overlap} and \ref{fig:conn} show nodes corresponding
to (finite) inductively successful S-reductions for $P_7$ and
$\mathtt{p(c)}$, and for $P_6$ and $\mathtt{conn(a,c)}$, respectively,
underlined.

If each $\theta_i$ is the external resolvent associated with the tree
transition $T_{i-1} \rightarrow T_i$ in an S-derivation $T_0 =
\kw{rew}(P, ? \gets t, \mathit{id}) \rightarrow T_1 \rightarrow
... \rightarrow T_n$, then $\theta_1,...,\theta_n$ is the {\em
  sequence of resolvents associated with} that S-derivation. In this
case, each tree $T_i$ in the S-derivation is given by $\kw{rew}(P,\, ?
\gets t, \,\theta_i\ldots \theta_2 \theta_1)$.  Note how the third
parameter composes the mgus.

\begin{example}
The S-derivation in Figure~\ref{fig:fibs2} starts with $\kw{rew}(P_3,?
\gets \mathtt{fibs(0,}$ $\mathtt{s(0),X)}, id)$.  Its second tree can
be seen as $\kw{rew}(P_3,? \gets \mathtt{fibs(0,s(0),X)}, \theta_1)$,
where $\theta_1 = \{\mathtt{X} \mapsto \mathtt{cons(0,S)}\}$, and its
third tree as $\kw{rew}(P_3,?  \gets \mathtt{fibs(0,s(0), X)},
\theta_2\theta_1)$, where $\theta_2 = \{\mathtt{Z} \mapsto
\mathtt{s(0)}\}$. Here, $\theta_1$ and $\theta_2$ are the resolvents
for the tree transitions for the first and the second trees with
respect to $X_3$ and $X_4$, respectively.
\end{example}

We have just formally rendered the formula \emph{Structural Logic
  Programming = S-Resolution Reductions + Control}: we embedded proof
search choices and or-node variable substitutions into S-resolution
reductions via rewriting trees, thus obtaining the notion of an
S-derivation and the inductive proof methodology we call {\em
  structural resolution}, or {\em S-resolution} for short.  It now
remains to exploit the inductive and coinductive properties of our new
theory of S-resolution.

\subsection{Inductive Soundness and Completeness of S-Resolution}

Before exploiting the coinductive properties of S-resolution we
investigate its inductive properties. Some S-derivations for a program
$P$ and a term $t$ may be S-refutations and some not, but termination
of one S-derivation in other than an inductive success tree does not
mean no S-refutation exists for $P$ and $t$. This reflects the facts
that inductive success is an existential property, and that entailment
for Horn clauses is only semi-decidable.  In this section we present
our inductive soundness and completeness results for S-resolution. We
note that these do not require logic programs to be either
observationally productive or non-overlapping.

\begin{example}\label{ex:succ}
An S-derivation for the program $P_6$ and $\mathtt{conn(a,c)}$ is
shown in Figure~\ref{fig:conn}.  The program $P_6$ is not
observationally productive. An inductive success subtree of the
derivation's final tree is indicated by underlining. It contains the
inductive success nodes labelled $\mathtt{conn(a,b)} \gets\;$ and
$\mathtt{conn(b,c)} \gets\;.$ Since its final tree is an inductive
success tree, this S-derivation is an S-refutation for $P_6$ and
$\mathtt{conn(a,c)}$.
\end{example}

\begin{example}
An S-refutation for the overlapping program $P_7$ and $\mathtt{p(c)}$
is shown in Figure~\ref{fig:overlap}. An inductive success subtree of
the derivation's final tree is indicated by underlining.
\end{example}

Inductive soundness and completeness of S-resolution are simple
corollaries of Theorem~\ref{th:isc}:

\begin{theorem}\label{thm:sc}
Let $P \in \cat{LP}(\Sigma)$ and $t \in \mathbf{Term}(\Sigma)$.
\begin{itemize}
\item (Inductive soundness of S-resolution) If there is an
  S-refutation for $P$ and $t$ that computes answer $\theta$, and $t'$ is  a ground instance 
	of $\theta(t)$, then $t' \in \mathbf{GTerm}(\Sigma)$.
\item (Inductive completeness of S-resolution) If $t \in M_P$, then
  there exists a term $t' \in \mathbf{Term}(\Sigma)$ that yields an
  S-refutation for $P$ and $t'$ that computes answer $\theta \in
  \cat{Subst}(\Sigma)$ such that $t$ is an instance of $\theta(t')$.
\end{itemize}
\end{theorem}

\noindent
We also have the following corollary of Theorem~\ref{th:isc}:

\begin{corollary}\label{cor:univ}
Let $P \in \cat{LP}(\Sigma)$ and $t \in \mathbf{Term}(\Sigma)$. If
there is an S-refutation $T_0 = \kw{rew}(P,?\gets t, \mathit{id})
\rightarrow T_1 \rightarrow \ldots \rightarrow T_n$ with associated
external resolvents $\sigma_1, \ldots, \sigma_n$ then, for all
grounding substitutions $\theta \in \cat{Subst}(\Sigma)$ for $\sigma_n
\ldots \sigma_1(t)$, $\theta \sigma_n \ldots \sigma_1(t)~\in~M_P$.
\end{corollary}

For an S-refutation $\kw{rew}(P,?\gets t, \mathit{id}) \rightarrow T_1
\rightarrow \ldots \rightarrow T_n$ with associated external
resolvents $\sigma_1, \ldots, \sigma_n$, the rewriting tree $T_n =
\kw{rew}(P, \ ?\gets t, \sigma_n \ldots \sigma_1)$ can be regarded as
a proof witness constructed for the query $t$.

The correspondence between the soundness and completeness of
S-refutations and the classical theorems of LP captures the
(existential) property of inductive success in S-resolution
reductions. Our results do not, however, mention failure, which is a
universal (and thus more computationally expensive) property to
establish.  Theorems~\ref{th:isc} and~\ref{thm:sc} also show that
rewriting trees can distinguish derivations proving logical entailment
existentially --- i.e., for some (ground) instances only --- from
those proving it universally --- i.e., for all (ground)
instances. Indeed, Theorems~\ref{th:isc} and Theorem~\ref{thm:sc} show
that proof search by unification has existential properties. 

\begin{example}\label{ex:conn2}
Since $\kw{rew}(P_1, ? \gets \mathtt{nat(X)}, \mathit{id})$ is not an
inductive success tree, $P_1$ does not logically entail the
universally quantified formula $\forall X.\,
\mathtt{nat}(X)$. Similarly, since $\kw{rew}(P_6, ?  \gets
\mathtt{conn(X,Y)}, \mathit{id})$ is not an inductive success tree,
$P_6$ does not logically entail $\forall X,Y.\,
\mathtt{conn}(X,Y)$. On the other hand, if we added a clause
$\mathtt{conn(X,X)} \gets$ to $P_6$, then, for resulting program
$P_6'$, $\kw{rew}(P'_6, ? \gets \mathtt{conn(X,X)}, \mathit{id})$
would be an inductive success tree, and we would be able to infer that
$P'_6$ does indeed logically entail $\forall X.\, \mathtt{conn}(X,X)$.
\end{example}

Throughout this section, finiteness of inductive success subtrees (and
thus of their corresponding rewriting reductions and S-derivations)
has served as a precondition for our inductive soundness and
completeness results. In the next section we restore the broken
symmetry by defining coinductive proof methods that require
observational productivity of S-derivations as a precondition of
coinductive soundness.
\section{Coinductive Soundness of S-Resolution}\label{sec:models} 

In this section, we show that S-resolution can capture not just
inductive declarative and operational semantics of LP, but coinductive
 semantics as well. We start by defining greatest complete
Herbrand models of logic programs, following~\cite{Llo87} closely,
then proceed by defining a notion of S-computations at infinity, and
conclude with a soundness theorem relating the two. We take time to
compare the computational properties of SLD-computations at infinity
and S-computations at infinity, and prove that the latter extends the
former. Since this section develops the theory of S-resolution for
coinductive LP, observational productivity is a necessary precondition
for establishing its results.

A first attempt to give an operational semantics corresponding to
greatest complete Herbrand models of logic programs was captured by
the notion of a \emph{computation at infinity} for SLD-resolution
\cite{EmdenA85,Llo87}. Computations at infinity are usually given
relative to an ultrametric on terms, constructed as follows:

\begin{definition}
A {\em truncation} for a signature $\Sigma$ is a mapping $\gamma':
\Nat \times \mathbf{Term}^\omega(\Sigma) \rightarrow
\mathbf{Term}(\Sigma \cup \diamond)$, where $\diamond$ is a new
nullary symbol not in $\Sigma$, and, for all $t \in
\mathbf{Term}^\omega(\Sigma)$ and $n \in \Nat$, the following
conditions hold:
\begin{itemize}
 \item $\mathit{dom}(\gamma'(n,t)) = \{m \in \mathit{dom}(t) \;| \ |m|
   \leq n \}$,
\item $\gamma'(n,t) = t(m)$ if $|m| < n$, and 
\item $\gamma'(n,t) = \diamond$ if $|m| = n$. 
\end{itemize}
\noindent	
For $t,s \in \mathbf{Term}^\omega(\Sigma)$, we define $\gamma(s,t) =
min\{n \;| \ \gamma'(n,s) \neq \gamma'(n,t)\}$, so that $\gamma(s,t)$
is the least depth at which $t$ and $s$ differ.  If we further define
$d(s,t) = 0$ if $s = t$ and $d(s,t) = 2^{-\gamma(s,t)}$ otherwise,
then $(\mathbf{Term}^\omega(\Sigma), d)$ is an ultrametric space.
\end{definition}

The definition of SLD-computable at infinity relative to a given
ultrametric is taken directly from~\cite{Llo87}:

\begin{definition}\label{def:SLD-comp-inf}
An SLD-resolution reduction is {\em fair} if either it is finite, or
it is infinite and, for every atom $B$ appearing in some goal in the
SLD-derivation, (a further instantiated version of) $B$ is chosen
within a finite number of steps. The term $t \in
\cat{GTerm}^\infty(\Sigma)$ is \emph{SLD-computable at infinity} with
respect to a program $P \in \cat{LP}(\Sigma)$ if there exist a $t' \in
\cat{Term}(\Sigma)$ and an infinite fair SLD-resolution reduction $G_0
= t', G_1, G_2, \ldots G_k \ldots$ with mgus $\theta_1, \theta_2,
\ldots \theta_k \ldots$ such that $d(t, \theta_k\ldots \theta_1(t'))
\rightarrow 0$ as $k \rightarrow \infty$. If such a $t'$ exists, we
say that $t$ is SLD-computable at infinity by $t'$.
\end{definition}

The fairness requirement ensures that infinite SLD-resolution
reductions that infinitely resolve against some subgoals while
completely ignoring others do not satisfy the definition of
SLD-computable at infinity. For example, $\mathtt{from (0,
  [0,[s(0),[s(s(0)), \ldots]]])}$ is not SLD-computable at infinity by
$P_5$ because no computation that infinitely resolves with subgoals
involving only $\mathtt{from}$ is fair.

In this section we see that \emph{SLD-Computations at Infinity =
  Global Productivity + Control}. Here, ``global productivity'' (as
opposed to observational productivity) requires that each fair
infinite SLD-resolution reduction for a program computes an infinite
term at infinity. The ``control" component determines the proof search
strategy for SLD-computations at infinity to be constrained by
fairness. We will see other variations on Kowalski's formula below.

Letting $P \in \mathbf{LP}(\Sigma)$ and defining $C_P = \{ t \in
\mathbf{GTerm}^\infty(\Sigma)\, |\, t$ is SLD-computable at infinity
with respect to $P$ by some $t' \in \cat{Term}(\Sigma)\}$, we have
that $C_P \subseteq M^{\omega}_P$ (\cite{EmdenA85,Llo87}).

\subsection{S-Computations at Infinity}

We can define a notion of computation at infinity for S-resolution to
serve as an analogue of Definition~\ref{def:SLD-comp-inf} for
SLD-resolution.  As a method of ``control" appropriate to
S-resolution, we introduce light typing for signatures, similar to
that in~\cite{GuptaBMSM07,SimonBMG07}. We introduce two types ---
namely, inductive and coinductive --- together with, for any signature
$\Sigma$, a {\em typing function} $\mathit{Ty}: \Sigma \rightarrow
\tau$ for $\Sigma$ that marks each symbol in $\Sigma$ as one or the
other. We adopt the convention that any symbol not explicitly marked
as coinductive is taken to be marked as inductive by default.  We note
that in SLD-computations at infinity all symbols are implicitly marked
as coinductive.

We extend the typing as inductive or coinductive from symbols to terms
and to nodes of rewriting trees. A term $t \in \mathbf{Term}(\Sigma)$
is inductive or coinductive according as $t(\epsilon)$ is.  If $P \in
\mathbf{LP}(\Sigma)$ and $T \in \mathbf{Rew}^\omega(P)$, then an
and-node $T(w)$ is coinductive if $T(w)(\epsilon)$ is coinductive, and
is inductive otherwise; an or-node in $T(w)$ is coinductive or
inductive according as its parent node is. A variable or-node $T(w) =
X$ is {\em open} if there exists a tree transition $T \rightarrow
T_X$, and is {\em closed} otherwise. A variable or-node is {\em
  coinductively open} if it is open and coinductive.  If $T'$ is a
rewriting subtree of $T$, then $T'$ is \emph{coinductively open} if it
contains coinductively open nodes, and is \emph{inductively closed} if
all of its open nodes are coinductive.

S-computations at infinity focus on observationally productive
programs and rely on properties of lightly typed rewriting trees. We
have:

\begin{definition}\label{df:infty}
Let $P \in \mathbf{LP}(\Sigma)$ be observationally productive, let
$Ty$ be a typing function for $\Sigma$, and let $t \in
\mathbf{GTerm}^\infty(\Sigma)$.  We say that $t' \in
\mathbf{Term}(\Sigma)$ {\em finitely approximates $t$ with respect to
  $P$ and $Ty$}, or is a {\em finite approximation of $t$ with respect
  to $P$ and $Ty$}, if the following hold:
\begin{enumerate}
\item There is an infinite S-derivation $T_0 = \kw{rew}(P,?  \gets t',
  \mathit{id}) \rightarrow T_1 \rightarrow \ldots T_k \rightarrow
  \ldots$ with associated resolvents $\theta_1, \theta_2, \dots
  \theta_k \ldots$ such that $d(t, \theta_k...\theta_1(t'))
  \rightarrow 0$ as $k \rightarrow \infty$.
\item This derivation contains infinitely many trees $T_{i_1},
  T_{i_2}, \ldots$ with an infinite sequence of corresponding
  rewriting subtrees $T'_{i_1}, T'_{i_2}, \ldots$ such that 
\begin{itemize}
\item[i)] each $T'_{i_j}$ is inductively closed and coinductively
  open
\item[ii)] each coinductive variable node is open and, for each such
  node $T'_{i_j}(w)$ in each $T'_{i_j}$, there exists $m > j$ such
  that $T'_{i_m}(wv)$ is coinductively open for some $v$.
\end{itemize}
\end{enumerate}
\noindent
Then $t$ is \emph{S-computable at infinity with respect to $P$ and
  $Ty$} if there is a $t' \in \mathbf{Term}(\Sigma)$ such that $t'$
finitely approximates $t$ with respect to $P$ and $Ty$. We define
$S^{Ty}_P = \{ t \in \mathbf{GTerm}^\infty(\Sigma) \, | \, $ $t$ is
S-computable at infinity with respect to $P$ and $Ty\}$.
\end{definition}

Here we see that \emph{S-Computations at Infinity = Global
  Productivity of S-Derivations + Control}. The first condition in
Definition~\ref{df:infty} ensures ``global productivity'' and the
second is concerned with ``control''. But Definition~\ref{df:infty}'s
requirement that programs are observationally productive is also used
to control S-derivations via observations.  We will see below that, as
the ``control" component becomes increasingly sophisticated, it can
capture richer cases of coinductive entailment than ever before.

\begin{example}\label{ex:cinf}
Consider $P_3$ and let $Ty$ be the type function marking (only) the
predicate $\mathtt{fibs}$ as coinductive. If $t' =
\mathtt{fibs(0,s(0),X)}$, then $t'$ finitely approximates, with
respect to $P_3$ and $Ty$, the infinite ground term $t^*$ from
Example~\ref{ex:fs} representing the stream of Fibonacci
numbers. Thus $t^*$ is S-computable at infinity with respect to $P_3$
and $Ty$. Figure~\ref{fig:fibs2} shows an initial fragment of the
S-derivation witnessing this. The infinite term $t^*$ is also
SLD-computable at infinity with respect to $P_3$.
\end{example}

Each of the ``control" requirements i, ii, and iii in
Definition~\ref{df:infty} is crucial to the correct formulation of a
notion of a finite approximation for S-resolution, and thus to the
notion of S-computability at infinity.  For Condition i, we note that
some S-derivations expand inductive nodes infinitely, which can block
the expansion of coinductive nodes. We do not want such S-derivations
to be valid finite approximations. For example, we want
$\mathtt{nats(scons(0,scons(0,\ldots)))}$ to be S-computable at
infinity with respect to $P_2$ if $\mathtt{nats}$ is marked
coinductive and $\mathtt{nat}$ is marked inductive, but we do not want
$\mathtt{nats(scons(s(s \ldots)),Y))}$ to be so computable.  Condition
i ensures that only S-derivations that infinitely expand only
coinductive nodes are valid finite approximations.

For Condition ii, we note that some S-derivations may have
unsuccessful inductive subderivations. We do not want these to be
valid finite approximations. For example, $P_5$ admits such
derivations.  Condition ii ensures that only S-derivations with
successful inductive subderivations are valid finite approximations.

For Condition iii, we note that even within one rewriting subtree
there may be several choices of coinductive nodes to expand in a
S-derivation. We want all such nodes to be infinitely expanded in a
valid finite approximation. For example, if $P_8$ comprises the
clauses of $P_2$ and $P_3$ with $\mathtt{fibs}$ and $\mathtt{nats}$
marked coinductive, together with $\mathtt{fibnats(X,Y)} \; \gets \;
\mathtt{fibs(0,s(0),X)}, \mathtt{nats(Y)}$, then S-derivations that
infinitely expand $\mathtt{fibs}$ but only finitely expand
$\mathtt{nats}$ compute at infinity terms of the form
$\mathtt{fibnats(cons(0,(cons(s(0),}$ $\mathtt{\ldots ))),
  scons}\mathtt{(}t_1,t_2\mathtt{))}$, for some finite terms $t_1$ and
$t_2$.  Since these computations do not expose the coinductive nature
of $\mathtt{nats}$, we do not want these to be valid finite
approximations. But we do want S-derivations that compute terms of the
form $\mathtt{fibnats(cons(0,(cons(s(0),\ldots)), scons(\_,scons(\_,}$
$\mathtt{\ldots)))}$ to be valid finite approximations.  Condition iii
ensures that only S-derivations infinitely expanding all coinductive
nodes are valid finite approximations.

\subsection{Soundness of S-Computations at Infinity}

We now investigate the relationship between SLD- and S-computations at
infinity. The next two examples show that, for a given $P \in
\mathbf{LP}(\Sigma)$ and a typing function $Ty$ for $\Sigma$, $C_P \,\subseteq \, S^{Ty}_P\,$ needs not hold.


\begin{example}
To see that $C_P \,\subseteq \,S^{Ty}_P$ needs not hold, we first note
that the infinite term $t = \mathtt{nat(s(s(\ldots)))}$ is
SLD-computable at infinity with respect to $P_1$ by $\mathtt{nat(X)}$,
and is thus in $C_{P_1}$.  But if $Ty$ marks $\mathtt{nat}$ as
inductive, then $t \not \in S^{Ty}_{P_1}$.  Similarly, in the mixed
inductive-coinductive setting we have that $t' =
\mathtt{nats(scons(s(s(\ldots)),scons(0, scons(s(0), \ldots)))}$ is
SLD-computable at infinity with respect to $P_2$ by
$\mathtt{nats(X)}$, and is thus in $C_{P_2}$. But if $Ty'$ is the
typing function that marks only $\mathtt{nats}$ as coinductive then,
since $\mathtt{nat}$ is (implicitly) marked as inductive, $t' \not \in
S^{Ty'}_{P_2}$.
\end{example}

Although for any specific typing function $Ty$ we need not 
have
$C_P \subseteq S_P^{Ty}$,
 considering all
typing functions simultaneously recovers a connection between
SLD-computability at infinity and S-computability at infinity.

\begin{definition}\label{def:S-hat}
If $P \in \mathbf{LP}(\Sigma)$ is observationally productive, then
$\widehat{S_P} \,= \,\bigcup\, \{S^{Ty}_P \,| \,Ty \textrm{ is }$
$\textrm{a typing function for } \Sigma \}$.
\end{definition}

The rest of this section formalises the relationship between $C_P$,
$\widehat{S_P}$, and $M^\omega_P$.

\begin{proposition}\label{prop:CP-in-SPhat}
Let $P \in \cat{LP}(\Sigma)$ be observationally productive. \\
 The infinite term $t \in
\cat{Term}^\infty(\Sigma)$ is SLD-computable at infinity by $t' \in
\cat{Term}(\Sigma)$ with respect to $P$\\
if and only if 
 there exists a typing
function $Ty$ for $\Sigma$ such that $t$ is S-computable at infinity
by $t'$ with respect to $P$ and $Ty$.
\end{proposition}

\begin{proof}
We must show that, for any $t \in \mathbf{GTerm}^\infty(\Sigma)$, if
$t$ is SLD-computable at infinity by $t'$ with respect to $P$, then
there is a typing function $Ty$ for $\Sigma$ such that $t$ is
S-computable at infinity by $t'$ with respect to $P$ and $Ty$. Since
$t$ is SLD-computable at infinity, there exist a $t' \in
\mathbf{Term}(\Sigma)$ and an infinite fair SLD-resolution reduction
$D$ of the form $G_0 = t' \rightarrow G_1 \rightarrow G_2 \ldots
\rightarrow G_k \rightarrow \ldots$ with mgus $\theta_1, \theta_2,
\ldots \theta_k \ldots$ such that $d(t, \theta_k \ldots \theta_1(t'))
\rightarrow 0$ as $k \rightarrow \infty$.

To show that $t$ is in $\widehat{S_P}$, consider $t'$, let $Ty$ be the
typing function marking all symbols in $\Sigma$ as coinductive.  We
construct an infinite S-derivation $D^*$ by first observing that each
SLD-resolution reduction step in $D$ proceeds either by matching or by
unification.  If $G_{i_1}$, $G_{i_2}$,.... is the sequence of lists in
$D$ out of which SLD-resolution reductions steps proceed by
unification, then let $D^*$ be the infinite S-derivation $T_0 =
\kw{rew}(P,?  \gets t', \mathit{id}) \rightarrow T_1 \rightarrow
\ldots T_j \rightarrow \ldots$, where $T_j = \cat{rew}(P, ? \gets t',
\theta_{i_j}...\theta_{i_1})$. We claim that $t$ is S-computable at
infinity with respect to $P$ and $Ty$ via the infinite S-derivation
$D^*$. The first condition of Definition~\ref{df:infty} is satisfied
because $d(t, \theta_k \ldots \theta_1(t')) \rightarrow 0$ as $k
\rightarrow \infty$ by the properties of $D$, and thus $d(t,
\theta_{i_j} \ldots \theta_{i_1}(t')) \rightarrow 0$ as $j \rightarrow
\infty$ by construction of $D^*$. To see that the second condition of
Definition~\ref{df:infty} is satisfied, recall that $D$ is fair and
infinite. Since $D$ is infinite and $Ty$ does not permit inductive
typing, $D^*$ contains (inductively closed and) coinductively open
rewriting trees infinitely often. As a result, $D^*$ satisfies i.
Since $Ty$ does not permit inductive typing, $D^*$ satisfies ii
trivially.  And $D^*$ satisfies iii because $D$ is both fair and
infinite.

In the opposite direction, suppose $D^* = T_0, T_1, \ldots$ is an infinite S-derivation that computes $t$ at infinity.
We need to show that there exists a corresponding $SLD$-derivation that is fair and non-failing.
It is easy to construct $D$ by following exactly the same resolvents as in $D^*$. We only need to show that such $D$ is fair and non-failing.
By definition, $T_0, T_1, \ldots$ should contain coinductive subtrees $T'_0, T'_1, \ldots$ in which every open 
coinductive node is resolved against
infinitely often. This means that corresponding derivation $D$ will be fair with respect to coinductively typed subgoals.
If the subtrees $T'_0, T'_1, \ldots$ do not involve inductive subgoals, then we have that $D$ is fair. (Because $D^*$ is non-terminating and non-failing, $D$ using the same resolvents will be nonfailing, too. )
Suppose $T'_0, T'_1, \ldots$ contained  inductive subgoals. 
By definition of $D^*$, every S-derivation step that resolved against a coinductive node is followed by a number of S-derivation steps that 
successfully close  all of the inductive subgoals in the corresponding rewriting subtrees. But that means that, in the corresponding
derivation $D$, these inductive subgoals will be chosen infinitely often and will not fail. This completes the proof. 
\end{proof}

We have the following immediate corollaries:

\begin{corollary}\label{cor:CP-in-SPhat}
If $P \in \mathbf{LP}(\Sigma)$ is observationally productive, then
$C_P = \widehat{S_P}$.
\end{corollary}


\begin{corollary}\label{th:si}
(Coinductive soundness of S-computations at infinity) If $P \in
  \mathbf{LP}(\Sigma)$ is observationally productive, then
  $\widehat{S_P} \subseteq M^{\omega}_P$.
\end{corollary}

\begin{proof}
Using Corollary~\ref{cor:CP-in-SPhat} and the fact that  $C_P \subseteq M^\omega_P$.
\end{proof}

Corollary~\ref{th:si} shows that, for observationally productive
programs, S-computations at infinity are sound with respect to
greatest complete Herbrand models. 
The corresponding completeness
result --- namely, that $M^{\omega}_P \subseteq \widehat{S_P}$ ---
does not hold, even if $P$ is observationally productive. The problem
arises when $P$ does not admit any infinite S-resolution
reductions. For example, if $P_{9}$ is the program with the single
clause $\mathtt{anySuccessor(s(X))} \gets$\;, then
$\mathtt{anySuccessor(s(s(\ldots )))} \in M^{\omega}_{P_{9}}$. But $P_{9}$
admits no infinite S-derivations, so no (infinite) terms are
S-computable at infinity with respect to $P_{9}$ and $Ty$ for any typing
function $Ty$.  A similar problem arises when $P$ fails the occurs
check. For example, if $P_{10}$ comprises the single clause
$\mathtt{p(X,f(X))} \; \gets \; \mathtt{p(X,X)}$, with $\mathtt{p}$
marked as coinductive, then $\mathtt{p(f(f(\ldots)),f(f(\ldots)))}$ is
in $M^{\omega}_{P_{10}}$ but is not S-computable at infinity with respect
to $P_{10}$ and $Ty$. This case is subtly different from the first one,
since $P_{10}$ defines the pair of  infinite terms $\mathtt{X=f(X)}$ only if
unification without the occurs check is permitted.

Corollary~\ref{th:si} ensures that (finite) coinductive terms logically
entail the infinite terms they finitely approximate. But there may, in
general, be programs for which coinductive terms also logically entail
other finite terms.

\begin{example}\label{ex:py}
Consider the program $P_{12}$ comprising the clause of $P_4$ and the
clause

\vspace*{0.1in}\noindent
1. $\mathtt{p(Y) \gets from(0,X)}$

\vspace*{0.1in}\noindent and suppose $Ty$ types only $\mathtt{from}$
as coinductive. Although $\mathtt{from(0,X)}$ finitely approximates an
infinite term with respect to $P_{11}$ and $Ty$, no infinite instance of
$\mathtt{p(Y)}$ is S-computable at infinity with respect to $P_{11}$ and
$Ty$. Nevertheless, $\mathtt{p(0)}$ and other instances of
$\mathtt{p(Y)}$ are logically entailed by $P_{11}$ and thus in
$M^\omega_{P_{11}}$.
\end{example}

\noindent
The following definition takes such situations into account:

\begin{definition}\label{def:ii}
Let $P \in \mathbf{LP}(\Sigma)$ be observationally productive, let
$Ty$ be a typing function for $\Sigma$, and let $t \in
\mathbf{Term}(\Sigma)$. Then $t$ is \emph{implied at infinity with
  respect to $P$ and $Ty$} if there exist terms $t_1, \ldots , t_n \in
\mathbf{GTerm}^\infty(\Sigma)$, each of which is S-computable at
infinity with respect to $P$ and $Ty$, and there exists a sequence of
rewriting reductions $t \rightarrow \ldots \rightarrow [t_1', \ldots,
  t_n']$ such that, for each $t_i$, $\theta(t_i') = t_i$ for some
$\theta \in \cat{Subst}^\omega(\Sigma)$.  We define $SI^{Ty}_P = \{ t
\in \mathbf{GTerm}^\omega(\Sigma) \, | \, $ $t$ is S-computable at
infinity or S-implied at infinity with respect to $P$ and $Ty\}$.
\end{definition}

\begin{example}
Consider once again the term $\mathtt{p(Y)}$ from Example~\ref{ex:py}
Although $\mathtt{p(Y)}$ is not {\em computable} at infinity with
respect to $P_{12}$ and $Ty$ as in Example~\ref{ex:py}, it is indeed {\em
  implied} at infinity with respect to $P_{12}$ and $Ty$.
\end{example}

Defining $\widehat{SI_P} \,= \,\bigcup\, \{SI^{Ty}_P \,| \,Ty \textrm{
  is }$ $\textrm{a typing function for } \Sigma \}$ gives the
following corollary of Corollary~\ref{th:si}:

\begin{corollary}\label{cor:si}
If $P \in \mathbf{LP}(\Sigma)$ is observationally productive, then
$\widehat{SI_P} \subseteq M^{\omega}_P$.
\end{corollary}


\section{Conclusions, Related Work, and Future Work}\label{sec:concl}

This paper gives a first complete formal account of the declarative
and operational semantics of structural (i.e., S-) resolution. We
started with characterisation of S-resolution in terms of big-step and
small-step operational semantics, and then showed that a rewriting
tree representation of this operational semantics is inductively sound
and complete, as well as coinductively sound.  Since observational
productivity is one of the most striking features of S-resolution,
much of this paper's discussion is centered around the subject of
productivity in its many guises: SLD-computations at infinity,
S-computations at infinity,  productive S-derivations with loop 
detection for rational terms (``observations of a coinductive proof") 
and sound observations of infinite derivations for irrational terms 
("sound observations").  We
have shown how an approach to productivity based on
S-resolution makes it possible to formalise the distinction between
global and observational productivity. This puts LP (and the broader
family of resolution-based methods) on par with coinductive methods in
ITP and TRS. We have also shown that our new notion of observational
productivity supports the formulation of a new coinductive proof
principle based on loop detection; moreover that proof principle is 
sound relative to S-computations at infinity and SLD-computations at 
infinity  known from the
1980s~\cite{Llo87,EmdenA85}. 

 The webpage \url{https://github.com/coalp} contains
implementation of sound observations of S-derivations and the Coq code supporting the proof-theoretic analysis 
of the loop detection method and its possible extensions.

The research reported herein continues the tradition of study of
infinite-term models of Horn clause
logic~\cite{JaffarS86,Llo87,EmdenA85,Jaume00}. In particular, we have
given a full characterisation of S-resolution relative to the least
and greatest fixed point semantics of LP, as is standard in the
classical LP literature.  Moreover, we have connected the classical
work on least and greatest complete Herbrand models of LP to the more
modern coalgebraic notation~\cite{Sg12} in Section~\ref{sec:cotrees}.
Our definitions of term trees and rewriting trees relate to the line
of research into infinite (term-)
trees~\cite{Courcelle83,JaffarS86,JohannKK15}.




\vspace*{0.05in}
S-resolution arose from coalgebraic studies of
LP~\cite{KPS12-2,KP11-2,KomendantskayaP11}, and these were
subsequently developed into a bialgebraic
semantics~\cite{BonchiZ13,BonchiZ15}. However, the bialgebraic
development takes the coalgebraic semantics of LP in a direction
different from our productivity-based analysis of
S-resolution. Investigating possible connections between observational
productivity of logic programs and their bialgebraic semantics offers
an interesting avenue for future work.

Another related area of research is the study of coinduction in first
order calculi other than Horn clause logic~\cite{BaeldeN12}, including fixed-point
linear logics (e.g.MuLJ)~\cite{ba08} and coinductive sequent
calculi~\cite{BrotherstonS11}. 
 One important methodological
difference between MuLJ (implemented as Bedwyr)~\cite{ba08} and
S-resolution is that Bedwyr begins with a strong calculus for
(co)induction and explores its implementations, while S-resolution
begins with LP's computational structure and constructs such a
calculus directly from it. Notably, Bedwyr requires cycle/invariant
detection, accomplished via heuristics that are incomplete but
practically useful. 
S-resolution 
may in the future provide further automation 
for systems like Bedwyr. 

The definition of observationally productive logic programs given in
this paper closely resembles the definitions of productive and guarded
corecursive functions in ITP --- particularly in Coq~\cite{BK08} and Agda, as
illustrated in Introduction.  

Further analysis of the
relationship between that coinductive proof principle and the one
developed here would require the imposition of a type-theoretic
interpretation on S-resolution. A type-theoretic view of S-resolution
for inductive programs is given in~\cite{FK15,FK16}. A preliminary
investigation of how coinductive hypothesis formation for Horn clauses
can be interpreted type-theoretically is given in~\cite{FKS15}.

Productivity has also become a well-established topic of research within
TRS community; see, e.g.,~\cite{EndrullisGHIK10,EndrullisHHP015}. The
definition of productivity for TRS relates to observational
productivity defined in this paper, and reflects the intuition of
finite observability of fragments of computations. However, because
S-resolution productivity is defined via \emph{termination} of
rewriting reductions, it also strongly connects to the termination
literature for TRS~\cite{Terese}.  Our definition in
Section~\ref{sec:cotrees} of S-resolution in terms of reduction
systems makes the connection between S-resolution and TRS explicit (see also~\cite{FK16}),
and thus encourages cross-pollination between research in S-resolution
and TRS.

The fact that productivity of S-resolution depends crucially on
termination of rewriting reductions makes this work relevant to
co-patterns~\cite{AbelPTS13}. In particular,~\cite{Basold2015}
considers a notion of productivity for co-patterns based on strong
normalisation of term-rewriting. This is similar to our notion of
observational productivity for logic programs. Further investigation
of applications of S-resolution in the context of co-patterns is under
way.

Observationally productive 
S-derivations may be seen as an example of \emph{clocked} corecursion~\cite{AtkeyM13}, where finite rewriting trees
give the measures of observation in a corecursive computation. Formal investigation of this relation is a future work.

Overall, we see the work presented here as laying a new foundation for
automated coinductive inference well beyond LP. In particular, we
expect our new methods to allow us to extend type inference algorithms
for a variety of programming
languages~\cite{AnconaLagorio11,Lammel:2005,AbelPTS13} to accommodate
richer forms of coinduction. We are currently exploring this enticing
new research direction.

\section{Acknowledgments}

We thank the following colleagues for discussions that encouraged and inspired this work:
Andreas Abel, Davide Ancona, Henning Basold,  Peng Fu, Gopal Gupta, Helle Hansen, Martin Hofmann and   Tom Schrijvers.
We particularly thank Vladimir Komendantskiy and Franti\^{s}ek Farka, who at different times implemented 
prototypes of CoAlgebraic Logic Programming (CoALP) and S-Resolution: their input has been invaluable for shaping this work.

\bibliographystyle{ACM-Reference-Format-Journals}
\bibliography{katya2}
\pagebreak
\appendix
\section{List of examples of Logic programs used across all sections}\label{sec:exp}

\small{
	\begin{tabular}{|p{1cm} p{6cm} p{6cm} |}
		\hline 
		Program short reference & Program clauses & Program meaning suggested by Herbrand models\\
		\hline \hline
		$P_1$ &

		$$ 0. \mathtt{nat(0)} \; \gets \; $$ 
				$$ 1. \mathtt{nat(s(X))} \; \gets \; \mathtt{nat(X)}$$

  & The set of all natural numbers\\
\hline
$P_2$ &  
$P_1$ and 
		$$ 2. \mathtt{nats(scons(X,Y))} \; \gets \; \mathtt{nat(X)},\mathtt{nats(Y)}$$
 & The set of natural numbers union the set of streams of natural numbers \\
\hline

$P_3$ & $$0. \mathtt{add(0,Y,Y)} \; \gets \;$$ 
$$ 1. \mathtt{add(s(X),Y,s(Z)) } \; \gets \; \mathtt{add(X,Y,Z)}$$
$$ 2. \mathtt{fibs(X,Y,cons(X,S)) } \; \gets \; \mathtt{add(X,Y,Z),
  fibs(Y,Z,S)}$$
	& The set of terms satisfying the relation of addition and terms denoting infinite streams of Fibonacci numbers \\ \hline
	
	$P_4$ & $$0. \mathtt{from(X, scons(X,Y))} \gets \mathtt{from(s(X),Y)}$$ & The set containing one term representing the infinite stream $0::s(0):: s(s(0))::\ldots$\\ \hline
		
	$P_5$ & $$0. \mathtt{from(X, scons(X,Y))} \gets \mathtt{from(s(X),Y)}, \mathtt{error(0)}$$ & The empty set\\ \hline
	$P_6$ &  $$0. \mathtt{conn(X,Y)} \gets \mathtt{conn(X,Z)}, \mathtt{conn(Z,Y)}$$
 $$1. \mathtt{conn(a,b)} \gets$$
 $$ 2. \mathtt{conn(b,c)} \gets$$ & The set $\{\mathtt{conn(a,b)}$ ,  $\mathtt{conn(b,c)}$,  $\mathtt{conn(a,a)}$,  $\mathtt{conn(a,b)}$,  $\mathtt{conn(b,c)}$,  $\mathtt{conn(b,b)}$,  $\mathtt{conn(c,c)}\}$\\ \hline
$P_7$ & $$0. \mathtt{p(c)} \gets$$
$$1. \mathtt{p(X)} \gets \mathtt{q(X)}$$ &  The set $\{\mathtt{p(c),q(c)}\}$. \\ \hline
$P_8$ & $P_2$, $P_3$ and 
$$ \mathtt{fibnats(X,Y)} \; \gets \;
\mathtt{fibs(0,s(0),X)}, \mathtt{nats(Y)} $$ & The union of sets for $P_2$ and $P_3$ plus all terms where predicate $\mathtt{fibnats}$ has  all terms given in models for  $P_2$ in the first argument and all terms given in models for $P_3$ in the second argument \\ \hline
		\end{tabular}}
		
		\small{
	\begin{tabular}{|p{1cm} p{6cm} p{6cm} |}
	\hline
$P_{9} $ & $$\mathtt{anySuccessor(s(X))} \gets$$ & The set $\{\mathtt{anySuccessor(s(0))}$, $\mathtt{anySuccessor(s(s((0))}$, $\ldots\}$ \\ \hline
$P_{10}$ & $$\mathtt{p(X,f(X))} \; \gets \; \mathtt{p(X,X)} $$ & The set $\{\mathtt{p(f(f(\ldots)),f(f(\ldots)))}\}$.\\ \hline
$P_{11}$ & $P_4$ and  $$\mathtt{p(Y) \gets from(0,X)}$$ & The set as for $P_4$ union the set of all terms $p(t)$, where $t$ is a term from the Herbrand base of that program.\\ \hline
$P_{12}$ & $$ \mathtt{zeros(scons(0,X))} \; \gets \; \mathtt{zeros(X)}$$ & The set contains one term -- denoting the stream of zeros. \\ \hline
		\end{tabular}}

\end{document}